\documentclass[publish]{smj}

\usepackage[utf8]{inputenc}
\usepackage{amsfonts} 
\usepackage{amsthm}
\usepackage{bbold}
\usepackage{setspace}
\usepackage{mathtools}
\usepackage{subcaption}

\usepackage{calc}
\usepackage{accents}

\newtheoremstyle{theorem}{12pt}{12pt}{\rm}{}{\sffamily}{ }{ }{}
\theoremstyle{theorem}

\newtheoremstyle{lemma}{12pt}{12pt}{\rm}{}{\sffamily}{ }{ }{}
\theoremstyle{lemma}

\newtheoremstyle{lemma*}{12pt}{12pt}{\rm}{}{\sffamily}{ }{ }{}
\theoremstyle{lemma*}
\newtheorem*{lemma*}{\sc Lemma}

\usepackage{scalerel}[2014/03/10]
\usepackage{stackengine, wasysym}

\newcommand\dwidetilde[1]{\ThisStyle{%
  \setbox0=\hbox{$\SavedStyle#1$}%
  \stackengine{+1.5\LMpt}{$\SavedStyle#1\,\,$}{%
    \stackengine{\dimexpr-4.5\LMpt+.3pt}{%
    \stretchto{\scaleto{\mkern2.5mu\thicksim}{.35\wd0}}{.25\ht0}%
    }{%
    \stretchto{\scaleto{\mkern2.5mu\thicksim}{.35\wd0}}{.25\ht0}%
    }{O}{c}{F}{T}{S}
  }{O}{c}{F}{T}{S}%
\!\!}}

\newcommand\dbwidetilde[1]{\ThisStyle{%
  \setbox0=\hbox{$\SavedStyle#1$}%
  \stackengine{+1.5\LMpt}{$\SavedStyle#1\,\,$}{%
    \stackengine{\dimexpr-4.5\LMpt+.3pt}{%
    \stretchto{\scaleto{\mkern-1.5mu\thicksim}{.35\wd0}}{.25\ht0}%
    }{%
    \stretchto{\scaleto{\mkern-1.5mu\thicksim}{.35\wd0}}{.25\ht0}%
    }{O}{c}{F}{T}{S}
  }{O}{c}{F}{T}{S}%
\!\!}}

\newcommand\dswidetilde[1]{\ThisStyle{%
  \setbox0=\hbox{$\SavedStyle#1$}%
  \stackengine{+1.5\LMpt}{$\SavedStyle#1\,\,$}{%
    \stackengine{\dimexpr-3.5\LMpt+.3pt}{%
    \stretchto{\scaleto{\mkern2.5mu\thicksim}{.35\wd0}}{.25\ht0}%
    }{%
    \stretchto{\scaleto{\mkern2.5mu\thicksim}{.35\wd0}}{.25\ht0}%
    }{O}{c}{F}{T}{S}
  }{O}{c}{F}{T}{S}%
\!\!}}

\newcommand\ttilde[1]{\ThisStyle{%
  \setbox0=\hbox{$\SavedStyle#1$}%
  \stackengine{+1.5\LMpt}{$\SavedStyle#1$}{%
    \stretchto{\scaleto{\SavedStyle\mkern2.5mu\thicksim}{.35\wd0}}{.25\ht0}%
  }{O}{c}{F}{T}{S}%
}}

\newcommand\bwidetilde[1]{\ThisStyle{%
  \setbox0=\hbox{$\SavedStyle#1$}%
  \stackengine{+1.5\LMpt}{$\SavedStyle#1$}{%
    \stretchto{\scaleto{\SavedStyle\mkern-1.5mu\thicksim}{.35\wd0}}{.25\ht0}%
  }{O}{c}{F}{T}{S}%
}}

\usepackage{calc}
\usepackage{accents}

\renewcommand{\widetilde}{\ttilde}

\usepackage{gensymb}

\Author{Elena Tuzhilina\Affil{1}, 
        Leonardo Tozzi\Affil{2}, 
        and Trevor Hastie\Affil{3}
}
\AuthorRunning{Elena Tuzhilina, Leonardo Tozzi and Trevor Hastie}

\Affiliations{

\item Department of Statistics, 
    Stanford University, 
    Stanford,
    CA, USA
    
\item Department of Psychiatry and Behavioral Sciences, 
    Stanford University, 
    Stanford,
    CA, USA
    
\item Department of Statistics, 
    Stanford University, 
    Stanford,
    CA, USA
}   

\CorrAddress{Elena Tuzhilina, 
            Department of Statistics, 
            Stanford University,
            Sequoia Hall,
            Stanford,
            CA-94305, USA}
\CorrEmail{elenatuz@stanford.edu}

\Title{\centering Canonical correlation analysis in high dimensions with structured regularization}
\TitleRunning{CCA in high dimensions with structured regularization}

\Abstract{
Canonical correlation analysis (CCA) is a technique for measuring the association between two multivariate data matrices. A regularized modification of canonical correlation analysis (RCCA) which imposes an $\ell_2$ penalty on the CCA coefficients is widely used in applications with high-dimensional data. One limitation of such regularization is that it ignores any data structure, treating all the features equally, which can be ill-suited for some applications. In this paper we introduce several approaches to regularizing CCA that take the underlying data structure into account. In particular, the proposed group regularized canonical correlation analysis (GRCCA) is useful when the variables are correlated in groups. We illustrate some computational strategies to avoid excessive computations with regularized CCA in high dimensions. We demonstrate the application of these methods in our motivating application from neuroscience, as well as in a small simulation example.
}


\begin{document}

\maketitle


\allowdisplaybreaks
\clearpage

\section{Introduction}
\label{intro}
Canonical correlation analysis (CCA) is a classic method commonly used in statistics for studying complex multivariate data. CCA was first introduced by \citet{cca0} as a tool for finding relationships between two sets of variables. It remains relevant in many domains including, but not limited to, genetics (see, for example, \citealp{genetics1}, \citealp{genetics2} and \citealp{genetics3}) and neuroscience (see, for example, \citealp{neuroscience1} and \citealp{neuroscience2}). In many applications the number of available observations is significantly smaller than the number of features under consideration, so some form of regularization is essential. Numerous regularized CCA extensions have been proposed (see, for example, \citealp{sparse1}, \citealp{sparse2} and \citealp{sparse3}). The most popular existing approach, called Regularized CCA (RCCA), imposes an $\ell_2$-penalty on the canonical coefficients (see, for example, \citealp{rcca1} and \citealp{rcca2}). Like any other standard regularization method based on the $\ell_2$ penalty, it has the property of treating all the coefficients equally and shrinking them toward zero. Although RCCA is well suited to data with general structure, in some applications the structure of the data can play an important role when investigating the association between variables. In this paper we develop several regularized extensions of CCA. These extensions were originally motivated by brain imaging applications, but the scope of applications can readily be extended to other fields.

The paper is organized as follows. In Section \ref{rcca} we introduce the necessary background for both CCA and RCCA methods. 
Next, we propose several approaches to regularization that account for the underlying structure of the data. In particular, in Section \ref{prcca} we introduce partially regularized canonical correlation analysis (PRCCA) that imposes an $\ell_2$ penalty only on a subset of the features. 
Although both RCCA and PRCCA problems have simple explicit solutions that can be computed via  singular value decomposition, they require us to work in terms of sample covariance matrices. This can be infeasible when the number of features is very large. In Sections~\ref{rcca:kernel} and \ref{prcca:kernel} we cover the ``kernel'' trick that allows to escape excessive computations while conducting CCA with regularization.
In Section \ref{grcca} we introduce group regularized CCA (GRCCA), a novel method that exploits group structure in the data. 

Since all the methods under consideration have similar structure, they can be considered as special cases of CCA with a general regularization penalty discussed in Sections \ref{genrcca}. All the technical details and proofs for these methods are covered in the Supplement.

We illustrate the proposed methods on our motivating study example involving functional brain imaging data in Section \ref{connectome} as well as on a small simulation example in Section \ref{simulation}. We conclude with a discussion, including some ideas for future work.

\section{Canonical correlation analysis with regularization}
\label{rcca}

\subsection{Canonical correlation analysis}
\label{cca}
Consider two random vectors $X\in\mathbb{R}^p$ and $Y\in\mathbb{R}^q$.
The goal of \textit{canonical correlation analysis (CCA)} is to find a linear combination of $X$ variables and a linear combination of $Y$ variables with the maximum possible correlation. Typically we find a sequence of such linear combinations.
Namely, for $i = 1,\ldots,\min(p,q)$ define a sequence of pairs of random variables $(U_i, V_i)$ as follows (see, for example, \citealp{cca2})
\begin{enumerate}
    \item Random variables $U_i$ and $V_i$ are linear combinations of $X$ and $Y$, respectively,~i.e.
    \begin{center}
        $U_i = \alpha_i^\top X~$  and $~V_i = \beta_i^\top Y.$
    \end{center}
    \item Coefficient vectors $\alpha_i\in\mathbb{R}^{p}$ and $\beta_i\in\mathbb{R}^{q}$ maximize the correlation
    $$\rho(\alpha_i, \beta_i) = \mathrm{cor}(\alpha_i^\top X, \beta_i^\top Y).$$
    \item Pair $(U_i, V_i)$ is uncorrelated with previous pairs, i.e
    \begin{center}
        $\mathrm{cor}(U_i, U_j)=\mathrm{cor}(V_i, V_j) = 0$ for $j<i.$
    \end{center}
\end{enumerate}
The pair $(U_i, V_i)$ is called \textit{$i$-th pair of canonical variates}; the corresponding optimal correlation value $\rho_i = ~\rho(\alpha_i, \beta_i)$ is called \textit{$i$-th canonical correlation}.

Note that correlation coefficient $\rho(\alpha, \beta)$ can be rewritten as 
\begin{align}
    \rho_{CCA}(\alpha, \beta) = \frac{\alpha^\top\Sigma_{XY}\beta}{\sqrt{\alpha^\top\Sigma_{XX}\alpha}~\sqrt{\beta^\top\Sigma_{YY}\beta}},
\label{cca:rho}
\end{align}
where $\Sigma_{XX}$, $\Sigma_{YY}$ and $\Sigma_{XY}$ refer to the covariance matrices $\mathrm{cov}(X)$, $\mathrm{cov}(Y)$ and $\mathrm{cov}(X,Y),$ respectively. It is easy to restate maximization of $\rho_{CCA}(\alpha, \beta)$ w.r.t. $\alpha$ and $\beta$ in terms of a constrained optimization problem  
\begin{align}
    \text{maximize} ~\alpha^\top\Sigma_{XY}\beta~ &\text{w.r.t.} ~\alpha\in\mathbb{R}^p~ \text{and} ~\beta\in\mathbb{R}^q \nonumber\\
    \text{subject to} ~\alpha^\top\Sigma_{XX}&\alpha = 1~ \text{and} ~\beta^\top\Sigma_{YY}\beta = 1.
\label{cca:op}
\end{align}
Thus, finding the $i$-th canonical pair is equivalent to solving the problem:  
\begin{align*}
    \text{maximize} ~\alpha_i^\top\Sigma_{XY}\beta_i~ &\text{w.r.t.} ~\alpha_i\in\mathbb{R}^p~ \text{and} ~\beta_i\in\mathbb{R}^q \\
    \text{subject to} ~\alpha_i^\top\Sigma_{XX}&\alpha_i = 1~ \text{and} ~\beta_i^\top\Sigma_{YY}\beta_i = 1\\
    ~\alpha_i^\top\Sigma_{XX}\alpha_j = 0~ \text{and}& ~\beta_i^\top\Sigma_{YY}\beta_j = 0~ \text{for} ~j<i.
\end{align*}
One can show that the canonical variates can be found via a singular value decomposition of the matrix $\Sigma_{XX}^{-\frac12}\Sigma_{XY}\Sigma_{YY}^{-\frac12}$, and that the canonical correlations coincide with the singular values of this matrix (see, for example, \citealp{cca1}). 

\subsection{Dealing with high dimensions}
\label{rcca:motivation}
In practice, we replace covariance matrices $\Sigma_{XX}$, $\Sigma_{YY}$ and $\Sigma_{XY}$ by the sample covariance matrices $\widehat\Sigma_{XX}$, $\widehat\Sigma_{YY}$ and $\widehat\Sigma_{XY}$. Specifically, suppose $\mathbf{X}\in\mathbb{R}^{n\times p}$ and $\mathbf{Y}\in\mathbb{R}^{n\times q}$ refer to matrices of $n$ observations for random vectors $X$ and $Y,$ respectively. Without loss of generality, assume that the columns of $\mathbf{X}$ and $\mathbf{Y}$ are centered (mean 0), then  
\begin{center}
$\widehat\Sigma_{XX} = \frac{1}{n}\mathbf{X}^\top \mathbf{X}$, $~\widehat\Sigma_{YY} = \frac{1}{n}\mathbf{Y}^\top\mathbf{Y}~$ and $~\widehat\Sigma_{XY}=\frac{1}{n}\mathbf{X}^\top\mathbf{Y}.$
\end{center}

If the number of observations $n$ is smaller than $p$ and/or $q$, the corresponding sample covariance matrices are singular and the inverses $\widehat\Sigma_{XX}^{-\frac12}$ and/or $\widehat\Sigma_{YY}^{-\frac12}$ do not exist.
\textit{Regularized canonical correlation analysis (RCCA)} resolves this problem by adding diagonal matrices to the sample covariance matrices of $X$ and~$Y$ (see, for example, \cite{rcca2} and \cite{rcca3}): 
\begin{align}
\widehat\Sigma_{XX}(\lambda_1) = \widehat\Sigma_{XX} + \lambda_1 I_p~ \text{ and } 
~\widehat\Sigma_{YY}(\lambda_2) = \widehat\Sigma_{YY} + \lambda_2 I_q.
\label{rcca:cov}
\end{align}
Here $I_p$ refers to the $p\times p$ identity matrix.
The \textit{modified correlation coefficient} that is maximized while seeking pairs of canonical variates is, hence,
\begin{align}
    \rho_{RCCA}(\alpha, \beta; \lambda_1, \lambda_2) = \frac{\alpha^\top\widehat\Sigma_{XY}\beta}{\sqrt{\alpha^\top(\widehat\Sigma_{XX}+\lambda_1 I)\alpha}~\sqrt{\beta^\top(\widehat\Sigma_{YY}+\lambda_2I)\beta}}.
\label{rcca:rho}
\end{align}
By analogy with CCA, it is easy to show that the RCCA variates can be found via the singular value decomposition of the matrix $(\widehat\Sigma_{XX} + \lambda_1 I)^{-\frac12}\widehat\Sigma_{XY}(\widehat\Sigma_{YY} + \lambda_2 I)^{-\frac12}$ and that RCCA modified correlations are equal to the singular values of this matrix.

\subsection{Shrinkage property}
\label{rcca:shrinkage}

Similar to ridge regression, regularization shrinks the CCA coefficients $\alpha$ and $\beta$ towards zero, where the \textit{penalty parameters} $\lambda_1$ and $\lambda_2$ control the strength of the shrinkage of $\alpha$ and $\beta$, respectively. This can be supported by the following reasoning.  
As in the case of CCA, maximization of the modified correlation $\rho_{RCCA}(\alpha, \beta;\lambda_1, \lambda_2)$ w.r.t. $\alpha$ and $\beta$ can be restated as a constrained optimization problem
\begin{center}
    maximize $~\alpha^\top\widehat\Sigma_{XY}\beta~$ w.r.t. $~\alpha\in\mathbb{R}^p~$ and $~\beta\in\mathbb{R}^q~$ \\
    subject to $~\alpha^\top(\widehat\Sigma_{XX}+\lambda_1 I)\alpha = 1~$ and $~\beta^\top(\widehat\Sigma_{YY}+\lambda_2I)\beta = 1.$
\end{center}
Note that the constraints can be rewritten as
\begin{center}
$\alpha^\top\widehat\Sigma_{XX}\alpha +\lambda_1 \|\alpha\|^2 = 1~$ and 
$~\beta^\top\widehat\Sigma_{XX}\beta +\lambda_1 \|\beta\|^2 = 1$
\end{center}
where $\|\cdot\|$ refers to the vector Euclidean norm. Finally, one can interpret $\lambda_1$ and $\lambda_2$ as Lagrangian dual variables for constraints $\|\alpha\|^2\leq t_1$ and $\|\beta\|^2\leq t_2$ which brings us to the optimization problem 
\begin{align}
    &\text{maximize} ~\alpha^\top\widehat\Sigma_{XY}\beta~ \text{w.r.t.} ~\alpha\in\mathbb{R}^p~ \text{and} ~\beta\in\mathbb{R}^q \nonumber\\
    \text{subject to} &~\alpha^\top\widehat\Sigma_{XX}\alpha = 1, ~\|\alpha\|^2 \leq t_1~ \text{and} ~\beta^\top\widehat\Sigma_{YY}\beta = 1, ~\|\beta\|^2\leq t_2.
\label{rcca:op}
\end{align}
One can show that for some appropriately chosen $t_1$ and $t_2$ this problem is equivalent to maximizing objective (\ref{rcca:rho}). Moreover, increasing $\lambda_1$ and $\lambda_2$ is equivalent to decreasing thresholds $t_1$ and $t_2$ which leads us to the shrinkage property. Finally, increasing $\lambda_1$ and $\lambda_2$ increases the denominator of (\ref{rcca:rho}) thereby shrinking the modified correlation coefficient to zero as well.

\subsection{RCCA kernel trick}
\label{rcca:kernel}
In some applications we need to deal with a very high-dimensional feature space. For instance, analyzing functional magnetic resonance imaging (fMRI) data, where the dimension refers to the number of brain regions (or voxels), the number of features can reach hundreds of thousands. If one of $p$ and $q$ is very large it can be problematic to store matrices $\widehat\Sigma_{XX}$ and $\widehat\Sigma_{YY}$. In this section we illustrate a simple trick based on the invariance of the RCCA problem under orthogonal transformations, that allows one to handle high-dimensional data when computing the RCCA solution. The idea to reduce a high dimensional CCA problem to a low-dimensional one via the kernel trick was previously introduced by \cite{kcca2} and \cite{kcca1}. Below we demonstrate the practical application of this idea to the RCCA problem that we subsequently use in the implementation. 

For simplicity, we assume that regularization is imposed on the $X$ part only, i.e. we assume $q<n$ and set $\lambda_2 = 0$. The same reasoning applies if we regularize $Y$ part as well. First we use the fact that any $n\times p$ matrix $\mathbf{X}$ with $p\gg n$ can be decomposed (e.g. via SVD) into a product $\mathbf{X} = \mathbf{R}V^\top$, where $\mathbf{R}\in \mathbb{R}^{n\times n}$ is a square matrix, and $V\in \mathbb{R}^{p\times n}$ is a matrix with orthonormal columns, i.e. $V^\top V = I$.

\begin{lemma*}[RCCA kernel trick]
The original RCCA problem stated for $\mathbf{X}$ and $\mathbf{Y}$ can be reduced to solving the RCCA problem for $\mathbf{R}$ and $\mathbf{Y}$. The resulting canonical correlations and variates for these two problems coincide. The canonical coefficients for the original problem can be recovered via the linear transformation $\alpha_X = V\alpha_R.$
\end{lemma*}

See Supplement Section \ref{supp:rcca:kernel} for the proof. Note that for $p\gg n$ the above trick allows us to avoid manipulating large $p\times p$ and $p\times q$ covariance matrices $\widehat\Sigma_{XX}$ and $\widehat\Sigma_{XY}$ and to operate in terms of smaller $n\times n$ and $n\times q$ matrices $\widehat\Sigma_{RR}$ and $\widehat\Sigma_{RY}$. Of course, exactly the same trick can be applied to the $Y$ part if $q\gg n.$

\subsection{Hyperparameter tuning}
\label{cv}
Before proceeding to our first example, let us discuss how one can tune the hyperparemeters. There are two hyperparameters for RCCA, i.e. $\lambda_1$ and $\lambda_2$. Let us denote the vector of hyperparameters by $\theta$. The values for these hyperparameters can be chosen via cross-validation. Below we present the outline for hold-out cross-validation; however, it can be naturally extended to the case of  $k$-fold cross-validation.

First we split all available observations into train $(\mathbf{X}_{train}, \mathbf{Y}_{train})$ and validation $(\mathbf{X}_{val}, \mathbf{Y}_{val})$ sets. We use the former set to fit the model and compute canonical coefficients $\alpha(\theta)$ and $\beta(\theta)$. Further, we use the latter set to estimate the model performance, i.e. we calculate 
$\rho_{val}(\theta) = \operatorname{cor}\left(\mathbf X_{val}\alpha(\theta), \mathbf Y_{val}\beta(\theta)\right).$ 
Note that here we utilize simple correlation instead of the modified correlation as a measure of performance. We pick the values of the hyperparameters maximizing the validation correlation, i.e.
$\theta^* = \operatorname{argmax}_\theta\left(\rho_{val}(\theta)\right),$ which can be done by means of grid search.

\section{Example. Human Connectome data study}
\label{connectome}

In this section we present an application of regularized CCA to data from a neuroscience study: the Human Connectome Project for Disordered Emotional States (HCP-DES) \cite{tozziHumanConnectomeProject2020}. One aim of HCP-DES is to link the function of macroscopic human brain circuits to self-reports of emotional well-being using magnetic resonance imaging. Here, we focused on brain activations during a Gambling task designed to probe the brain circuits underlying reward (described in detail in \cite{barchFunctionHumanConnectome2013, tozziHumanConnectomeProject2020}). We linked this neuroimaging data with self-reports assessing various aspects of reward-related behaviors (Behavioral Approach System/Behavioral Inhibition Scale (BIS/BAS), \cite{carverBehavioralInhibitionBehavioral1994}), depression symptoms (Mood and Anxiety Symptom Questionnaire (MASQ), \cite{wardenaarDevelopmentValidation30item2010a}) and positive as well as negative affective states (Positive and Negative Affect Schedule (PANAS), \cite{watsonDevelopmentValidationBrief1988a}). We selected participants who had complete self-report and imaging data as well as no quality control issues, for a total of 153 participants (94 females, 59 males, mean age 25.91, sd 4.85). For details on the preprocessing and subject-level modeling used to derive brain activations in response to the Gambling task, see Section \ref{supp:rde:preprocessing} of the Supplement. We used for our analysis the activations for the monetary reward compared to monetary loss during the task. For each subject, the activation at each greyordinate (grey matter coordinate) in the brain was extracted, yielding a matrix $\mathbf{X}$ of $n = 153$ rows (subjects) and $p = 90~368$ columns (greyordinates).
The self-report data consisted of 9 variables: drive, fun seeking, reward responsiveness (from the BAS), total behavioral inhibition (from the BIS), distress, anhedonia, anxious arousal (from the MASQ), positive and negative affective states (from the PANAS). These were entered in a matrix $\mathbf{Y}$ of $n=153$ rows (subjects) and $q=9$ columns.

\begin{figure}[h!]
  \begin{subfigure}[b]{0.5\textwidth}
    \includegraphics[width=\textwidth]{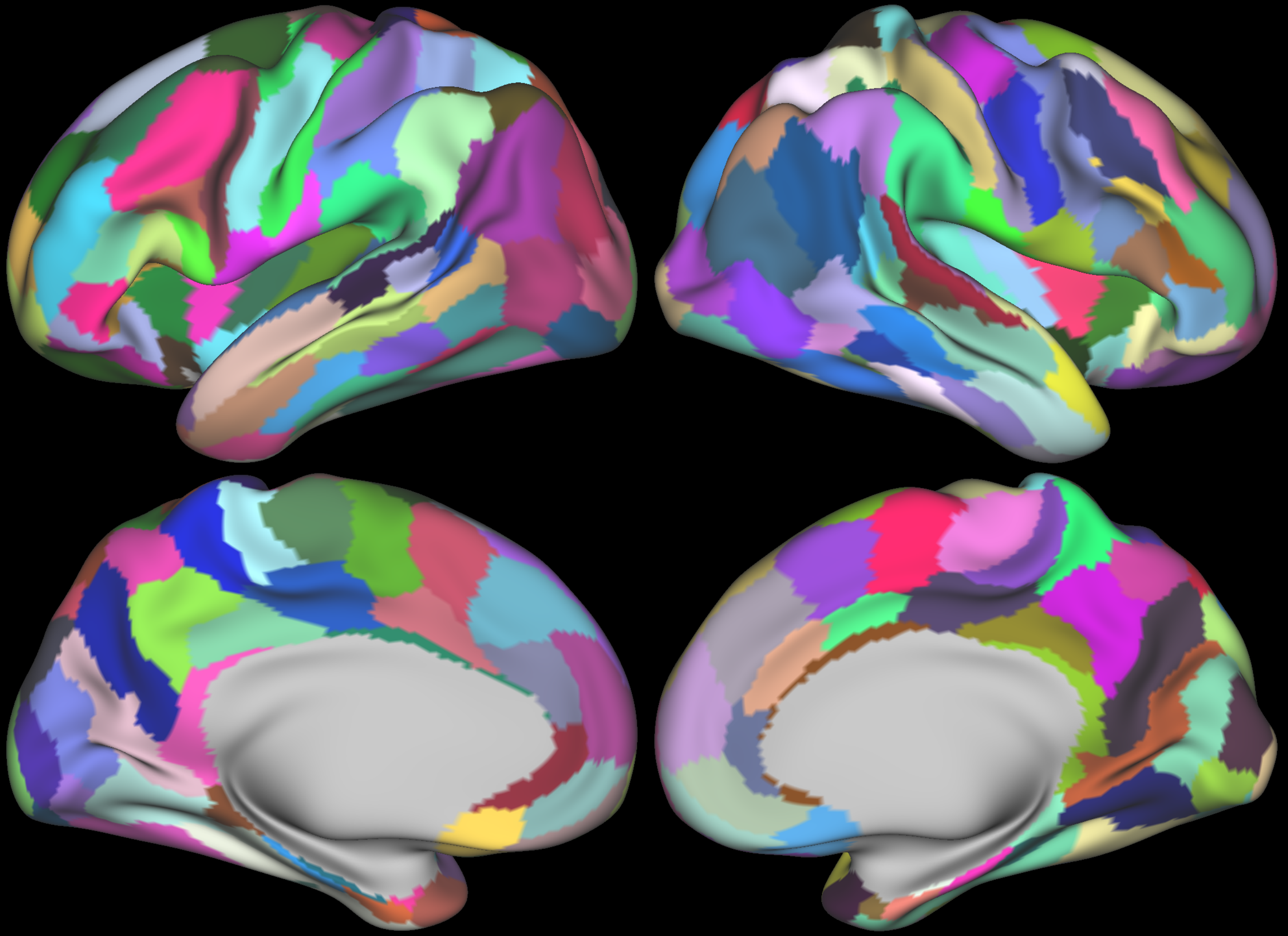}
    \subcaption{Cortical parcellation.}
     \end{subfigure}
  \hfill
  \begin{subfigure}[b]{0.5\textwidth}
    \includegraphics[width=\textwidth]{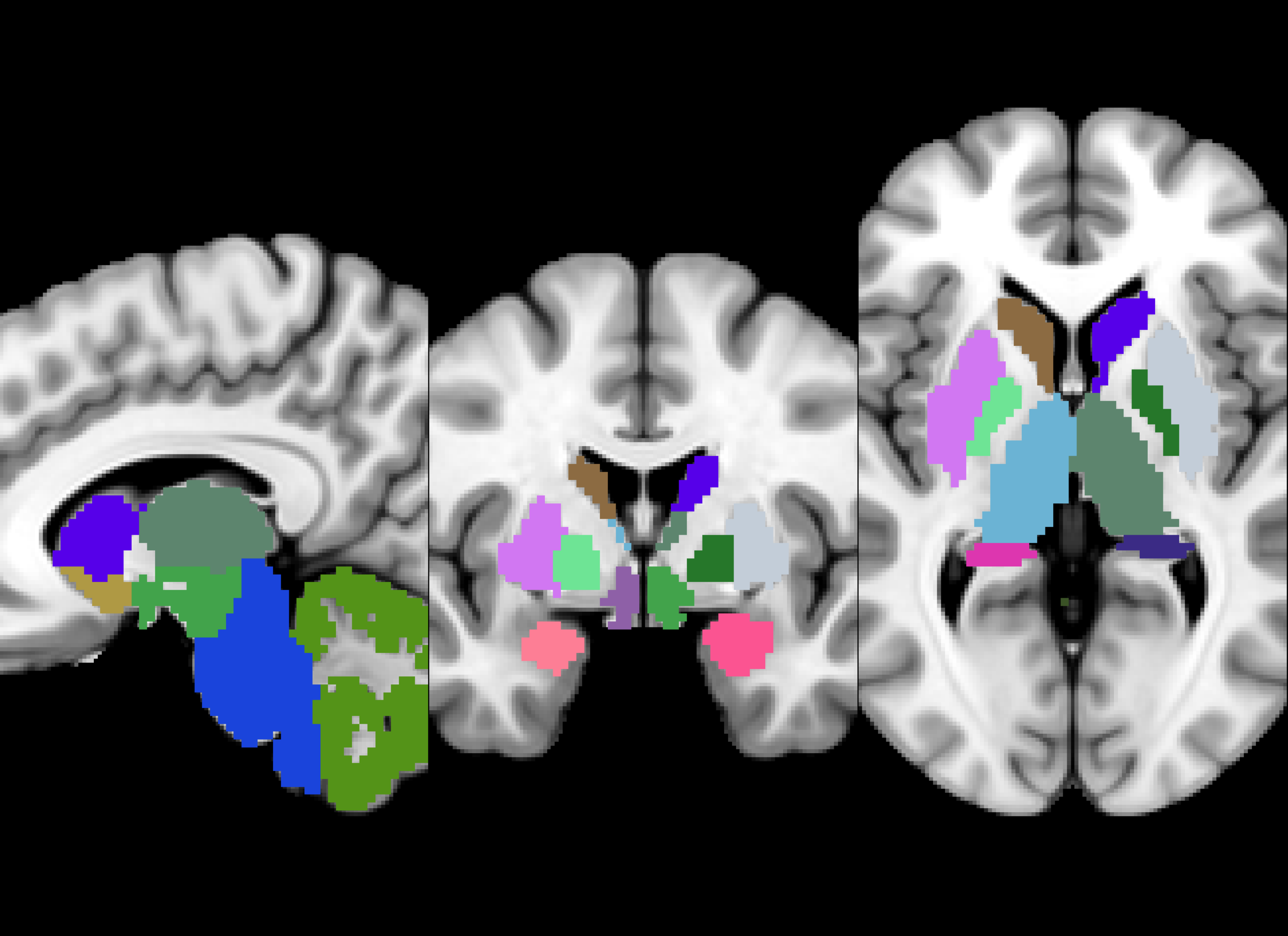}
    \subcaption{Subcortical parcellation.}
    \end{subfigure}
  \caption{229 brain regions: 210 cortical regions of the Brainnetome atlas and 19 subcortical regions of the Desikan-Killiany Atlas.}
  \label{fig:rde:parcellation}
\end{figure}

To test our structured methods, the $90~368$ greyordinates were grouped into 229 regions depending on their functional and anatomical properties, corresponding to the 210 regions of the Brainnetome atlas \citep{fanHumanBrainnetomeAtlas2016a} with the addition of the 19 subcortical regions of the Desikan-Killiany Atlas \citep{desikanAutomatedLabelingSystem2006a}. The resulting brain regions are presented in Figure~\ref{fig:rde:parcellation}. 

\subsection{RCCA results}
\label{connectome:rcca}
We start with a relatively simple model averaging activation inside each brain region and using the averaged values as features. Thus  $\mathbf{X}$ becomes of dimension $153\times 229$. 
To remove the effect of sex on the resulting correlation we adjusted both activation and behavioral data for the binary sex variable by means of simple linear regression (mean adjustment). Since the $\mathbf{Y}$ matrix is relatively small we imposed no penalty on $\mathbf{Y}$ ($\lambda_2 = 0$). To pick the hyperparameter $\lambda_1$ we ran 10-fold cross-validation on the adjusted data pairs with the penalty factor varying over the grid $\lambda_1 = 10^{-3}, 10^{-2}, \ldots, 10^4, 10^5$.

The resulting cross-validation curves represent the unpenalized correlation between canonical variates (computed on the 9 folds of train set and one fold of validation set) averaged across 10 folds (see Figure \ref{fig:rcca:mean:scores}). Note that although larger $\lambda_1$ shrinks the modified correlation $\rho_{RCCA}$ toward zero, the unpenalized correlation is not guaranteed to be monotonically decreasing in $\lambda_1$. According to the plot, the highest score is achieved for $\lambda_1 = 0.1$ with the corresponding test correlation equal to~$0.148$. 
Using the kernel trick we now run RCCA for the original activation data ($90~368$ features adjusted for the sex effect). According to Figure \ref{fig:rcca:scores}, the maximum score is equal to $0.11$ ($\lambda_1 = 0.001$).

\begin{figure}[h!]
  \centering
  \begin{subfigure}[b]{\textwidth}
    \includegraphics[width = 0.97\textwidth]{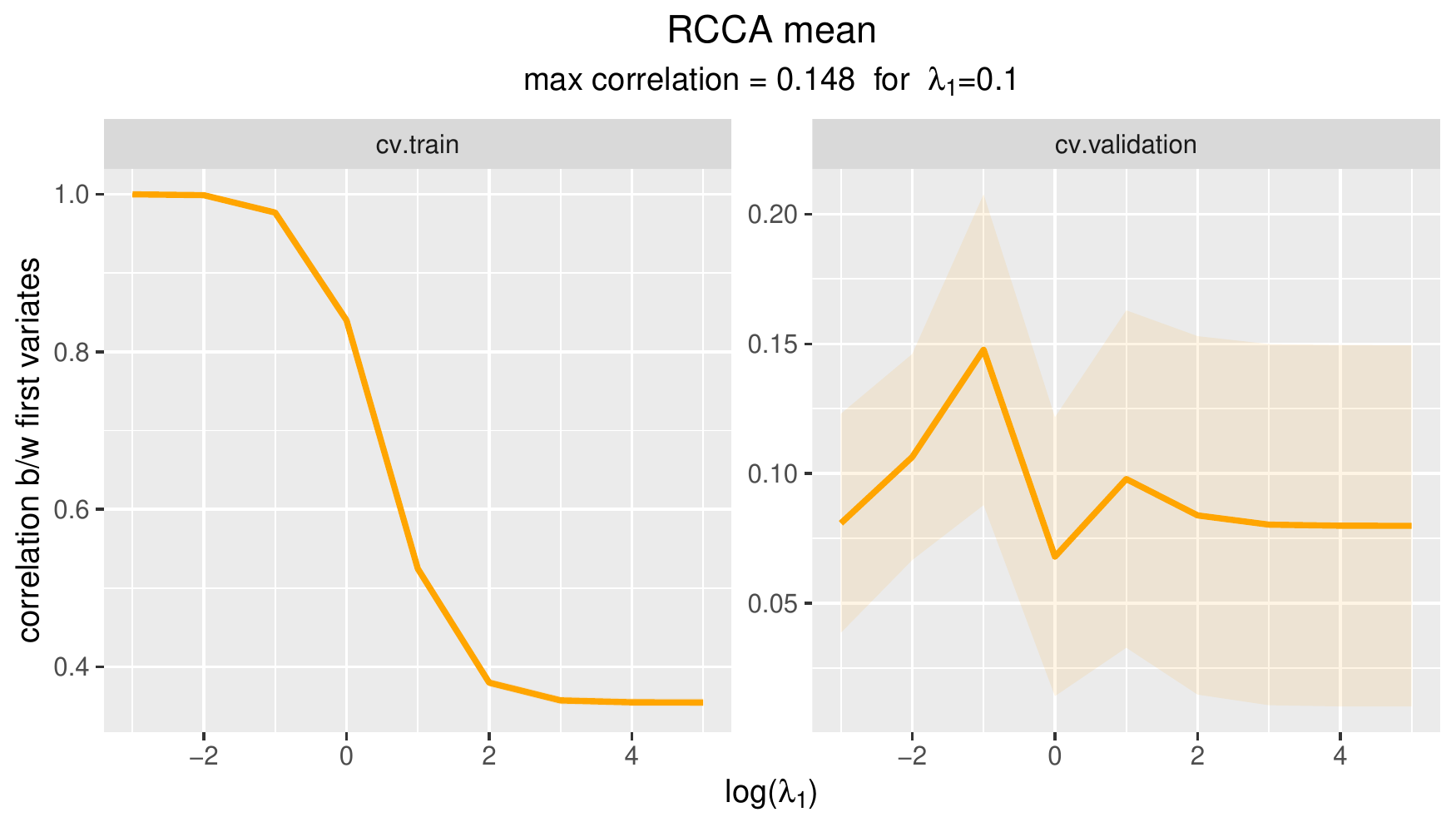}
    \subcaption{Averaged data}
      \label{fig:rcca:mean:scores}
    \end{subfigure}
    \vfill
     \begin{subfigure}[b]{\textwidth}
    \includegraphics[width = 0.97\textwidth]{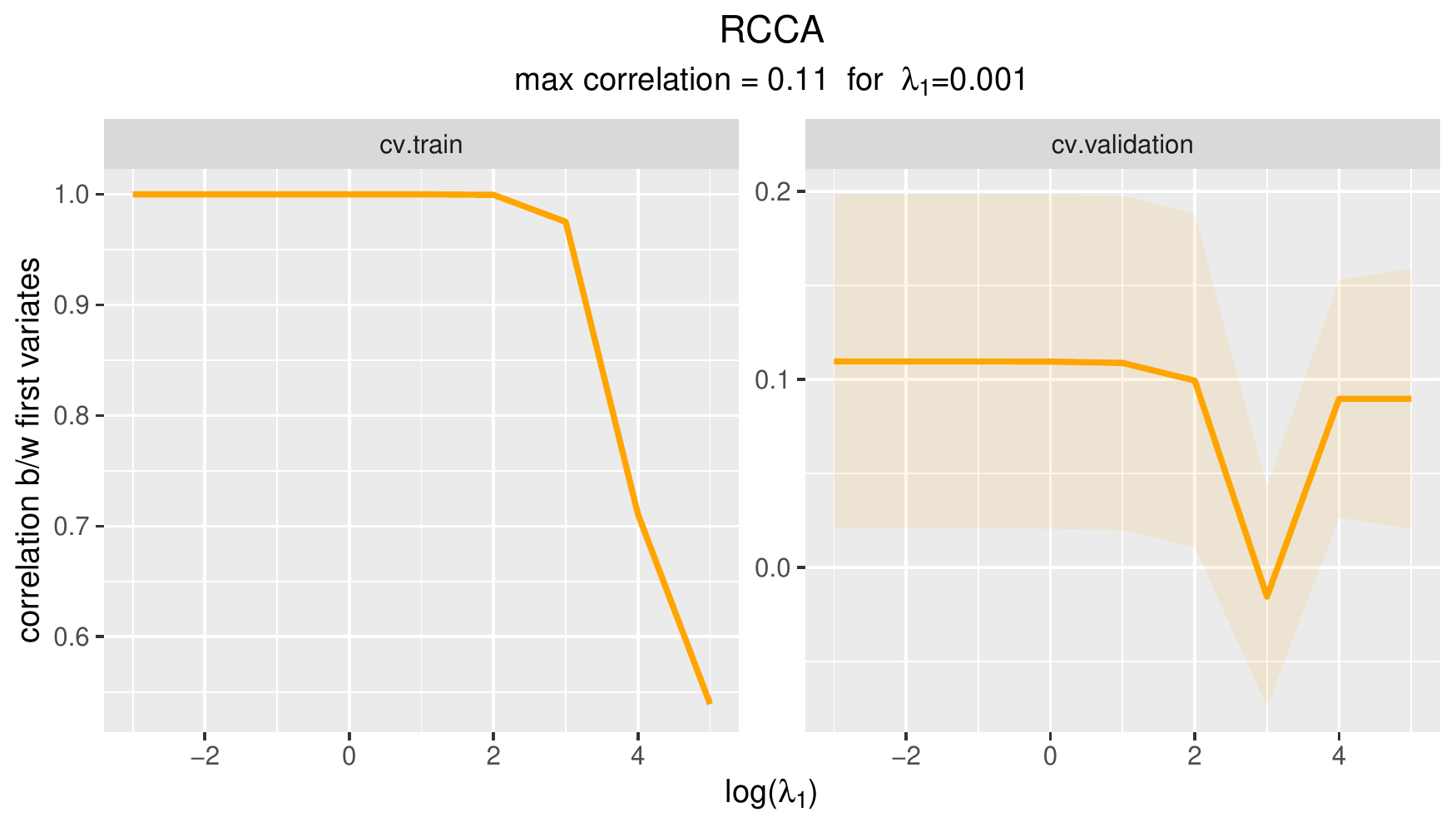}
    \subcaption{Original data}
      \label{fig:rcca:scores}
     \end{subfigure}
 \caption{The cross-validation curves obtained via RCCA with 10-fold cross-validation for the Human Connectome dataset. Left panel: (unpenalized) correlation between train canonical variates. Right panel: (unpenalized) correlation between validation canonical variates.}
\end{figure}

In general, to compare two models and check that the cross-validation score does not reflect spurious findings we should validate the performance of the models on an independent test set. However, the small sample size of the data (only 153 observations) makes the test correlation estimates unreliable. 
We use nested cross-validation (NCV) to overcome the problem of overfitting to the dataset that we use for tuning. 
Specifically, we split the data in 11 folds. Each of the 11 folds is given an opportunity to be used as an independent test set, while all other 10 folds folds are used to tune the hyperparameters via 10-fold cross-validation. Therefore, we report 11 cross-validation scores along with 11 test scores, and we present the average for both as well as 1SE confidence intervals (see Figure \ref{fig:rcca:ncv}). 

According to the NCV scores, the cross-validation procedure for the full RCCA model discovered significant correlation (independent test set correlation averaged across 11 folds is 0.105). However, the correlation value obtained by the smaller mean RCCA model was way too optimistic (average test score is 0.044 with wide error bands).

It is worth noting the computational speed of the proposed method. Unlike the \texttt{rcc()} function from the popular \texttt{CCA} R package (see \citet{rcca3}), which is not able to handle such a large number of features ($\approx 90$K for the $X$ side), our implementation of RCCA with the kernel trick completes the calculations in 20 seconds.

\subsection{Interpretetability of canonical coefficients}
\label{rcca:interpretability}

In this section we visualize the RCCA coefficients $\alpha$ corresponding to the hyperparameters chosen by cross-validation. Recall that the original $X$ features represent the brain activation detected at each brain greyordinate, so we can map the resulting RCCA coefficients back to the brain surface (see Figure \ref{fig:rde:rcca:coef}). There is an apparent trade-off between the interpretability and flexibility of the model. Namely, although full data RCCA is more flexibile, there is quite a lot of variation in the resulting canonical coefficients. This makes the corresponding brain image harder to interpret. On the other hand, the reduced model allows us to identify the brain regions that have the highest impact on the resulting correlation. However, it loses in terms of flexibility (and, potentially, performance). In what follows, we aim to develop the model that links these two extremes enabling us to control the interpretability vs. flexibility trade-off. 

\begin{figure}[h!]
  \begin{subfigure}[b]{0.5\textwidth}
    \includegraphics[width=\textwidth]{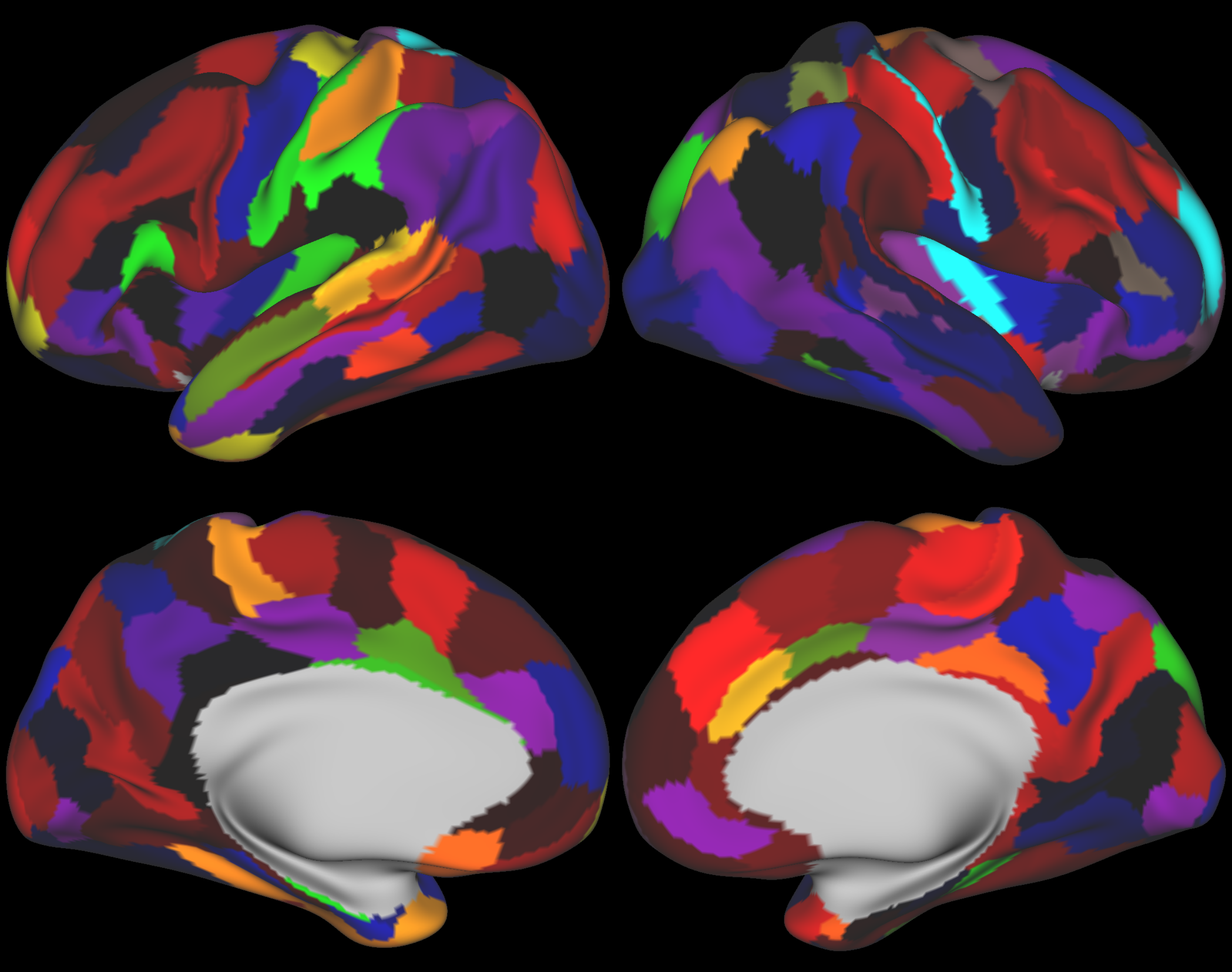}
    \subcaption{Cortical coefficients, averaged data.}
    \label{fig:rde:rcca:coef:a}
     \end{subfigure}
  \hfill
  \begin{subfigure}[b]{0.5\textwidth}
    \includegraphics[width=\textwidth]{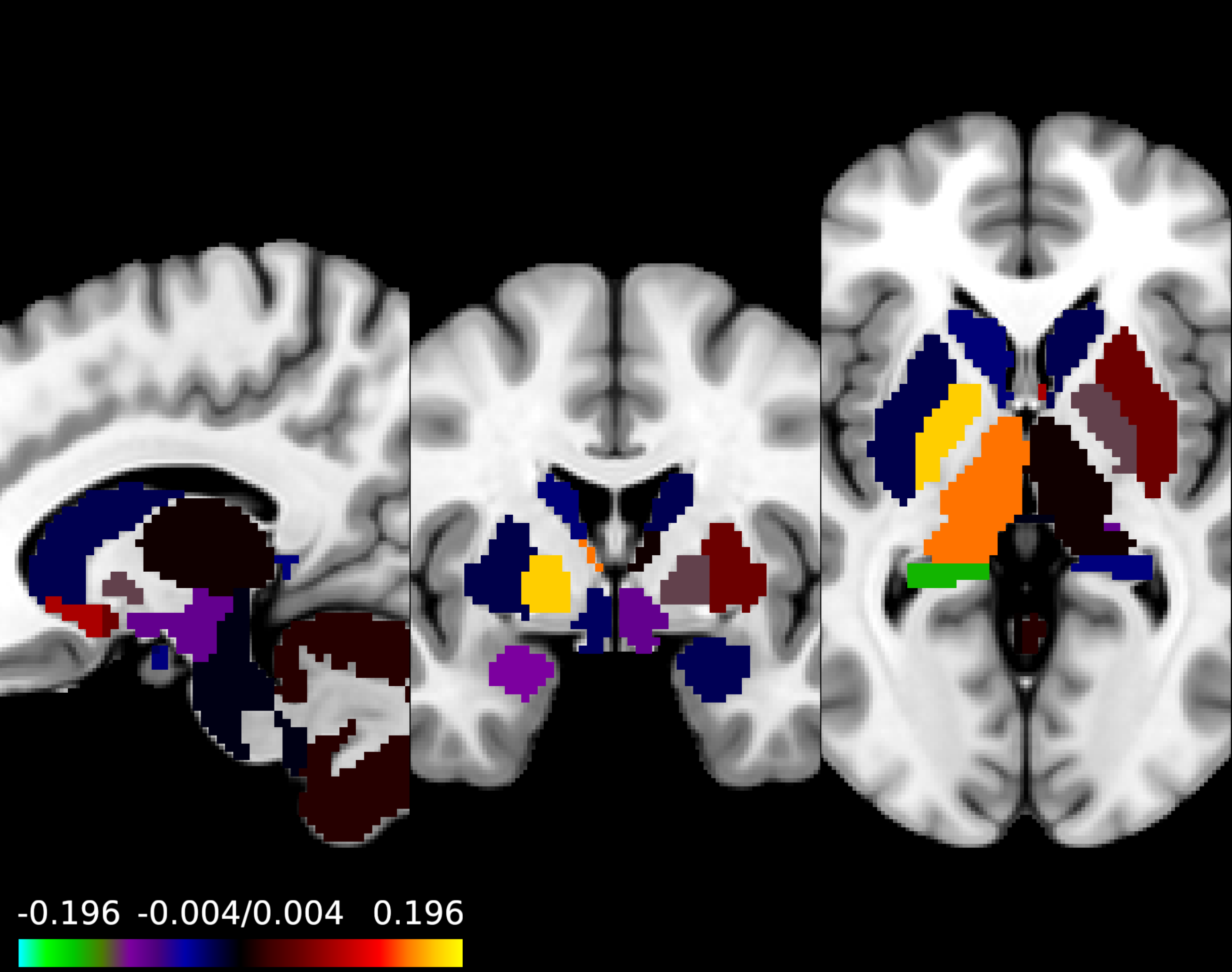}
    \subcaption{Subcortical coefficients, averaged data.}
      \label{fig:rde:rcca:coef:b}
    \end{subfigure}
    \vfill
  \begin{subfigure}[b]{0.5\textwidth}
    \includegraphics[width=\textwidth]{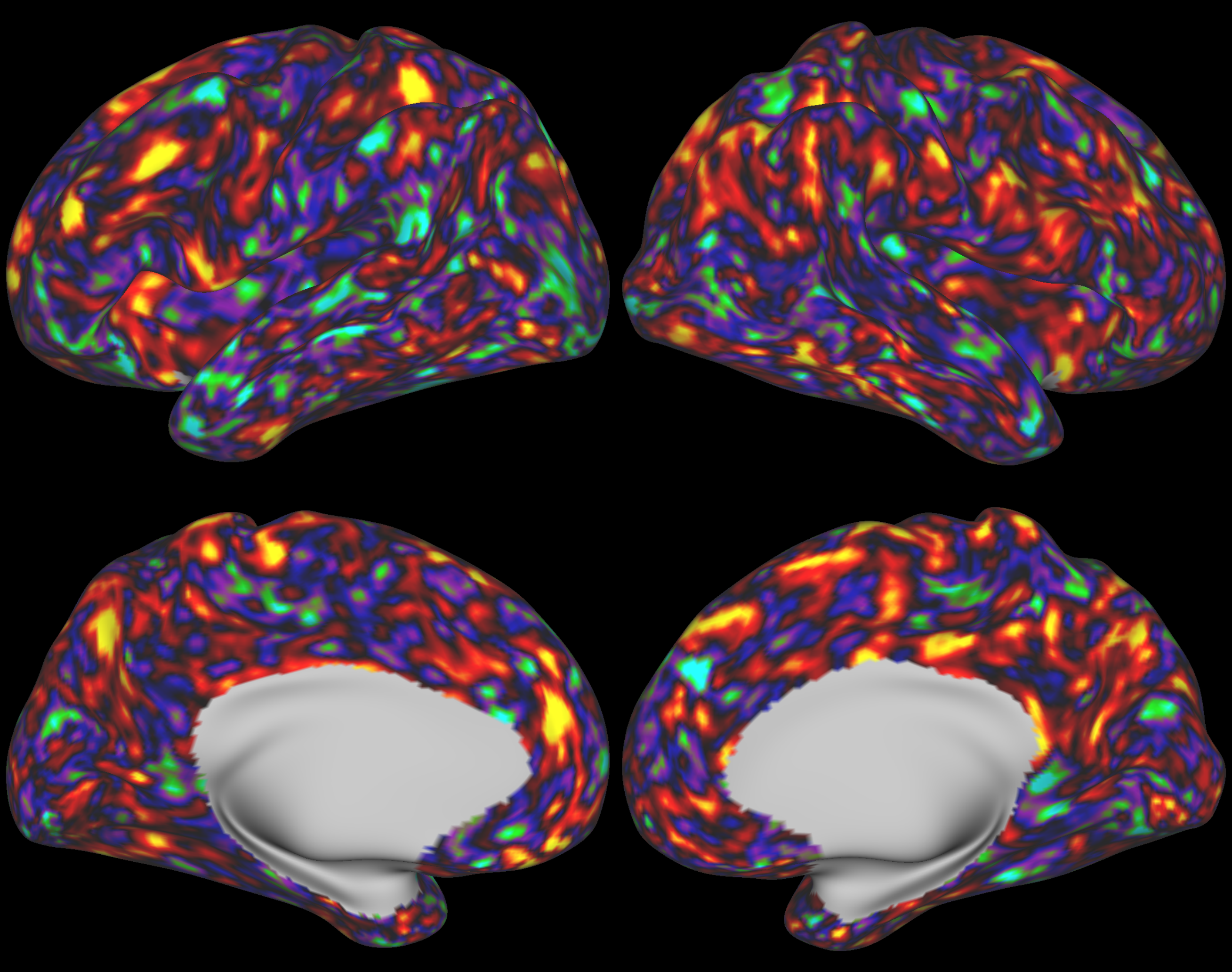}
    \subcaption{Cortical coefficients, original data.}
    \label{fig:rde:rcca:coef:c}
     \end{subfigure}
  \hfill
  \begin{subfigure}[b]{0.5\textwidth}
    \includegraphics[width=\textwidth]{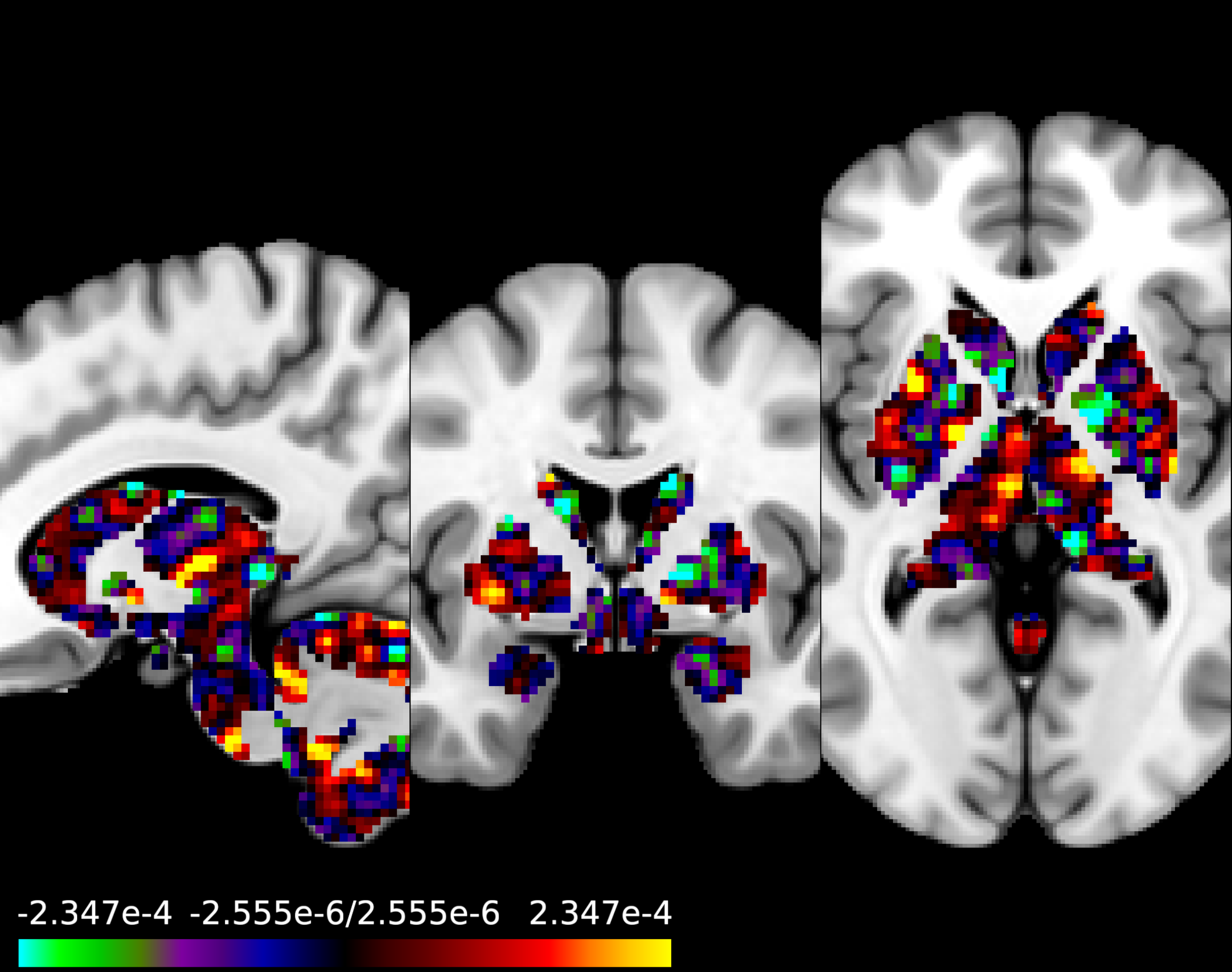}
    \subcaption{Subcortical coefficients, original data}
    \label{fig:rde:rcca:coef:d}
    \end{subfigure}
    
\caption{RCCA coefficients computed for the averaged data ($229$ features) with hyperparameter $\lambda_1 = 0.1$ chosen by cross-validation and for the original data ($90~368$ features) with hyperparameter $\lambda_1 = 10$ chosen by cross-validation.}
  \label{fig:rde:rcca:coef}
\end{figure}

\section{CCA with partial regularization}
\label{prcca}
\subsection{Penalizing a subset of canonical coefficients}
\label{prcca:motivation}
Suppose you are interested in the influence of a specific brain region on the resulting CCA correlation, however, you do not want to completely eliminate the remaining brain regions from the data. Recall that the inequality constraints in the RCCA optimization problem (\ref{rcca:op}) control the deviation of all canonical coefficients from zero. PRCCA is a modification of RCCA that allows one to shrink only a subset of the CCA coefficients leaving the complement unpenalized. The proposed PRCCA method is a key building block for our final group RCCA approach and also admits some independent interesting applications.  

Suppose that both $\alpha$ and $\beta$ are split into two parts
\begin{center}
$\alpha = \begin{psmallmatrix}
\alpha_1\\
\alpha_2
\end{psmallmatrix}$ 
where $\alpha_1\in\mathbb{R}^{p_1}$ and $\alpha_2\in\mathbb{R}^{p_2}$ with $p_1+p_2 = p,$ \\
$\beta = \begin{psmallmatrix}
\beta_1\\
\beta_2
\end{psmallmatrix}$ 
where $\beta_1\in\mathbb{R}^{q_1}$ and $\beta_2\in\mathbb{R}^{q_2}$ with $q_1+q_2 = q.$ 
\end{center}
Replacing the constraints $\|\alpha\|^2\leq t_1$ and $\|\beta\|^2\leq t_2$ in the optimization problem (\ref{rcca:op}) by $\|\alpha_1\|^2\leq t_1$ and $\|\beta_1\|^2\leq t_2$, respectively, we get the PRCCA optimization problem 
\begin{align}
    &\text{maximize} ~\alpha^\top\widehat\Sigma_{XY}\beta~ \text{w.r.t.} ~\alpha\in\mathbb{R}^p~ \text{and} ~\beta\in\mathbb{R}^q \nonumber\\
    \text{subject to} &~\alpha^\top\widehat\Sigma_{XX}\alpha = 1, ~\|\alpha_1\|^2 \leq t_1~ \text{and} ~\beta^\top\widehat\Sigma_{YY}\beta = 1, ~\|\beta_1\|^2\leq t_2.
\label{prcca:op}
\end{align}
Re-expressing the constraints for $\alpha$ and $\beta$ in terms of dual variables we get
\begin{align*}
\alpha^\top\widehat\Sigma_{XX}\alpha + \lambda_1\|\alpha_1\|^2 =\alpha^\top\left(\widehat\Sigma_{XX} + \lambda_1
\begin{psmallmatrix}
I_{p_1} & 0\\
0 & 0
\end{psmallmatrix} 
\right)\alpha\\
\beta^\top\widehat\Sigma_{YY}\beta + \lambda_2\|\beta_1\|^2 = \beta^\top\left(\widehat\Sigma_{YY} + \lambda_2
\begin{psmallmatrix}
I_{q_1} & 0\\
0 & 0
\end{psmallmatrix} 
\right)\beta.
\end{align*}
This leads us to the PRCCA modification of the correlation coefficient as follows
\begin{align}
    \rho_{PRCCA}(\alpha, \beta;\lambda_1, \lambda_2) = \frac{\alpha^\top\widehat\Sigma_{XY}\beta}{\sqrt{\alpha^\top(\widehat\Sigma_{XX}+\lambda_1 \begin{psmallmatrix}
I_{p_1} & 0\\
0 & 0
\end{psmallmatrix})\alpha}~\sqrt{\beta^\top(\widehat\Sigma_{YY}+\lambda_2 \begin{psmallmatrix}
I_{q_1} & 0\\
0 & 0
\end{psmallmatrix})\beta}}.
\label{prcca:rho}
\end{align}
As usual, PRCCA variates and coefficients obtained by maximizing (\ref{prcca:rho}) can be found by means of SVD of the matrix $\left(\widehat\Sigma_{XX} + \lambda_1 \begin{psmallmatrix}
I_{p_1} & 0\\
0 & 0
\end{psmallmatrix} \right)^{-\frac12}\widehat\Sigma_{XY}\left(\widehat\Sigma_{YY} + \lambda_2 \begin{psmallmatrix}
I_{q_1} & 0\\
0 & 0
\end{psmallmatrix} \right)^{-\frac12}.$

\subsection{PRCCA kernel trick}
\label{prcca:kernel}
In this section we extend the kernel trick to the PRCCA problem set up. Note that, because $I$ was replaced by block matrix
$\begin{psmallmatrix}
I & 0\\
0 & 0
\end{psmallmatrix}$ in the denominator of the modified correlation coefficient (\ref{rcca:rho}), the PRCCA problem does not preserve the property of invariance under orthogonal transformations. Thus, the mathematics used in Section \ref{rcca:kernel} does not work anymore. 
There are two main ingredients for the PRCCA kernel trick. First, if the feature matrix consists of two orthogonal blocks, then the kernel trick can be applied to each block independently. Second, there exists a non-orthogonal transformation of the feature matrix making the two blocks orthogonal to each other, while resulting in an equivalent PRCCA problem.

We again assume for simplicity that the regularization penalty is imposed on the $X$ part only.  
Suppose $X = \begin{psmallmatrix}
X_1\\
X_2
\end{psmallmatrix},$
where random vectors $X_1\in\mathbb{R}^{p_1}$ and $X_2\in\mathbb{R}^{p_2}$ correspond to penalized coefficients part $\alpha_1$ and unpenalized part $\alpha_2$,  respectively. Let $\mathbf{X}_1\in\mathbb{R}^{n \times p_1}$ and $\mathbf{X}_2\in\mathbb{R}^{n \times p_2}$ represent the corresponding matrices of observations,  so $\mathbf{X} =~(\mathbf{X}_1,\mathbf{X}_2)$. To make the PRCCA solution identifiable we require $\mathbf{X}_2$ to be tall and full rank, i.e. $p_2 < n$ and $\operatorname{rank}(\mathbf{X}_2) = p_2.$ We can also assume that $p_1\gg n.$

As the first step we find a linear transformation $A\in\mathbb{R}^{p \times p}$ such that matrix $\bwidetilde{\mathbf{X}} = \mathbf{X} A = (\bwidetilde{\mathbf{X}}_1, \bwidetilde{\mathbf{X}}_2)\in\mathbb{R}^{n \times p}$
has orthogonal blocks $\bwidetilde{\mathbf{X}}_1\in\mathbb{R}^{n \times p_1}$ and $\bwidetilde{\mathbf{X}}_2\in\mathbb{R}^{n \times p_2}$, i.e. $\bwidetilde{\mathbf{X}}_1^\top\bwidetilde{\mathbf{X}}_2~=~0,$ and that preserves the second block, i.e. $\bwidetilde{\mathbf{X}}_2 = \mathbf{X}_2$. This can be easily done by linear regression (see Supplement for details). This linear transformation maps the original PRCCA problem to an equivalent one (equivariant in terms of the coefficients and invariant in terms of the objective). 

Note that the above transformation forces the sample covariance matrix $\widehat\Sigma_{\widetilde{X}\widetilde{X}}$ to be block-diagonal with blocks $\widehat\Sigma_{\widetilde{X}_1\widetilde{X}_1}$ and $\widehat\Sigma_{\widetilde{X}_2\widetilde{X}_2}$, which enables us to apply the kernel trick to the first and the second block of $\bwidetilde{\mathbf{X}}$ independently. Specifically, consider the decomposition of $\bwidetilde{\mathbf{X}}_1$, as in the RCCA kernel lemma, i.e. $\bwidetilde{\mathbf{X}}_1 = \mathbf{R}_1V_1^\top$ where $V_1\in\mathbb{R}^{p_1\times n}$ is a matrix with orthonormal columns and $\mathbf{R}_1\in\mathbb{R}^{n\times n}$ is some square matrix. Then the following lemma holds.  

\begin{lemma*}[PRCCA kernel trick]
The original PRCCA problem stated for $\mathbf{X}$ and $\mathbf{Y}$ can be reduced to solving the smaller PRCCA problem for  $\mathbf{R} =~\begin{psmallmatrix}
\mathbf{R}_1\\
\mathbf{X}_2
\end{psmallmatrix}$ and $\mathbf{Y}$. 
The resulting canonical correlations and variates for these two problems coincide. The canonical coefficients for the original problem can be recovered via the linear transformation $\alpha_X = A\begin{psmallmatrix}
V_1 & 0\\
0&I
\end{psmallmatrix}\alpha_{R}.$
\end{lemma*}

See Supplement Section \ref{supp:prcca:kernel} for the proof. Note that, according to the lemma, instead of working with large matrices $\widehat\Sigma_{XX}\in\mathbb{R}^{p\times p}$ and $\widehat\Sigma_{XY}\in\mathbb{R}^{p\times q}$ we can operate in terms of smaller matrices $\widehat\Sigma_{RR}\in\mathbb{R}^{(n + p_2)\times (n+p_2)}$ and $\widehat\Sigma_{R Y}\in\mathbb{R}^{(n + p_2)\times q}$ thereby avoiding excessive computations.

\subsection{Testing PRCCA on the Human Connectome data}
\label{connectome:prcca}
First we chose the brain region of interest that we aim to release from the regularization penalty. To do so for each column of averaged activation data we compute \textit{Cohen's d}, which measures the effect size for a one-sample t-test comparing the population mean to zero, and pick the regions with at least a medium effect ($d>0.3$). The resulting $26$ regions demonstrate the largest activation during the Gambling task. Then we run PRCCA on the averaged data imposing the penalty on all but these $26$ regions. We consider the same grid of values $\lambda_1 = 10^{-3}, 10^{-2}, \ldots, 10^4, 10^5$ and choose the hyperparemeter according to the maximum cross-validation score (see Section \ref{cv} for the details). The highest validation correlation is equal to $0.073$ and is achieved when $\lambda_1 = 1$ (see Figure \ref{fig:prcca:scores}). Figure \ref{fig:rde:prcca:coef} represents the canonical coefficients computed for $\lambda_1$ chosen by cross-validation. As expected, the coefficients for most regions were shrunk to zero leaving only a few standing out.
Again, it is the kernel trick which enables running cross-validation for the extremely high-dimensional feature matrix in just a few minutes. Although PRCCA does not perform well in this application, it will play an important role in developing subsequent methods. 

\begin{figure}[h!]
  \centering
    \includegraphics[width = \textwidth]{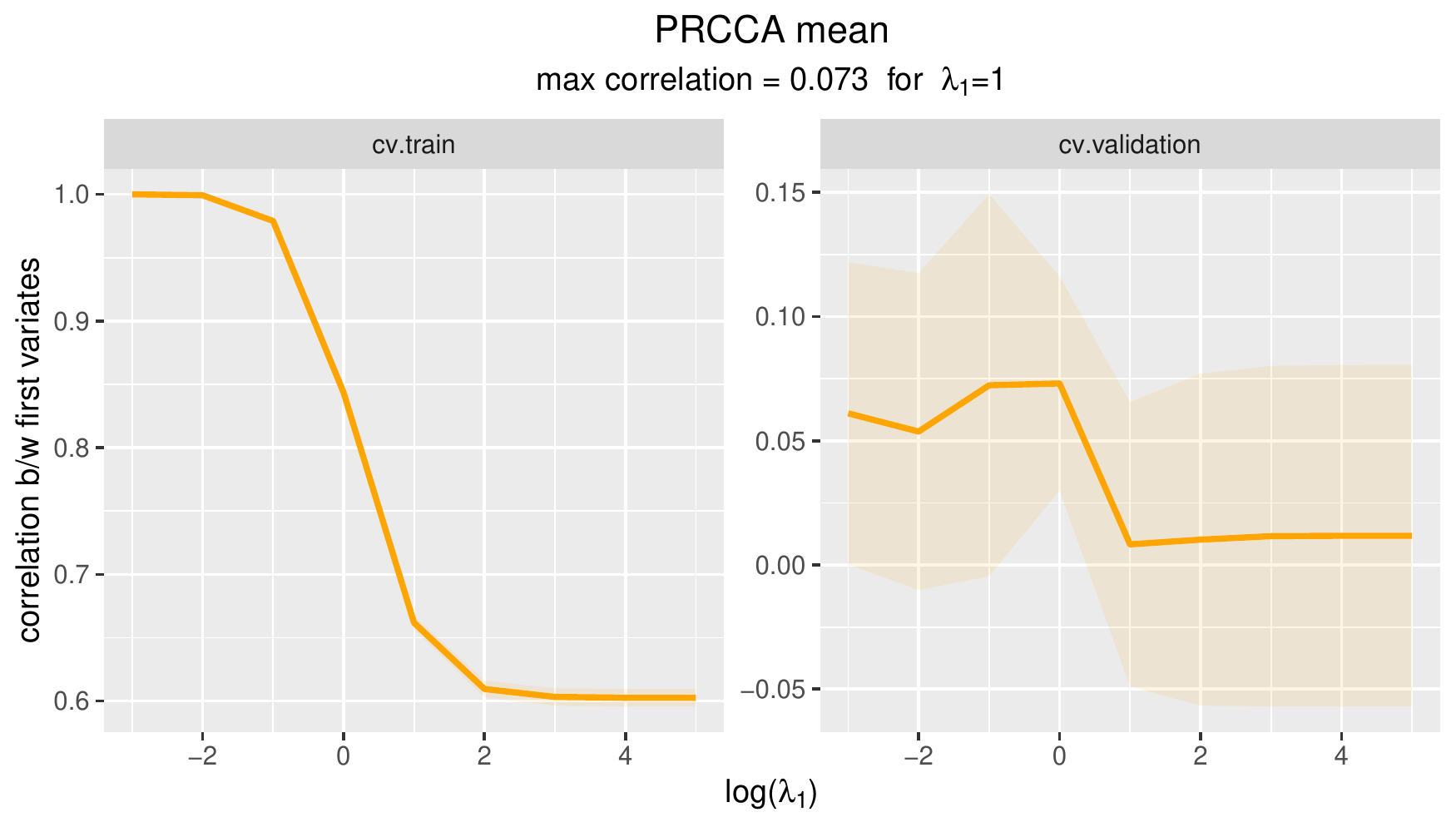}
\caption{The cross-validation curves obtained via PRCCA with 10-fold cross-validation for the Human Connectome dataset. Left panel: (unpenalized) correlation between train canonical variates. Right panel: (unpenalized) correlation between validation canonical variates.} 
\label{fig:prcca:scores}
\end{figure}

\begin{figure}[h!]
  \begin{subfigure}[b]{0.5\textwidth}
    \includegraphics[width=\textwidth]{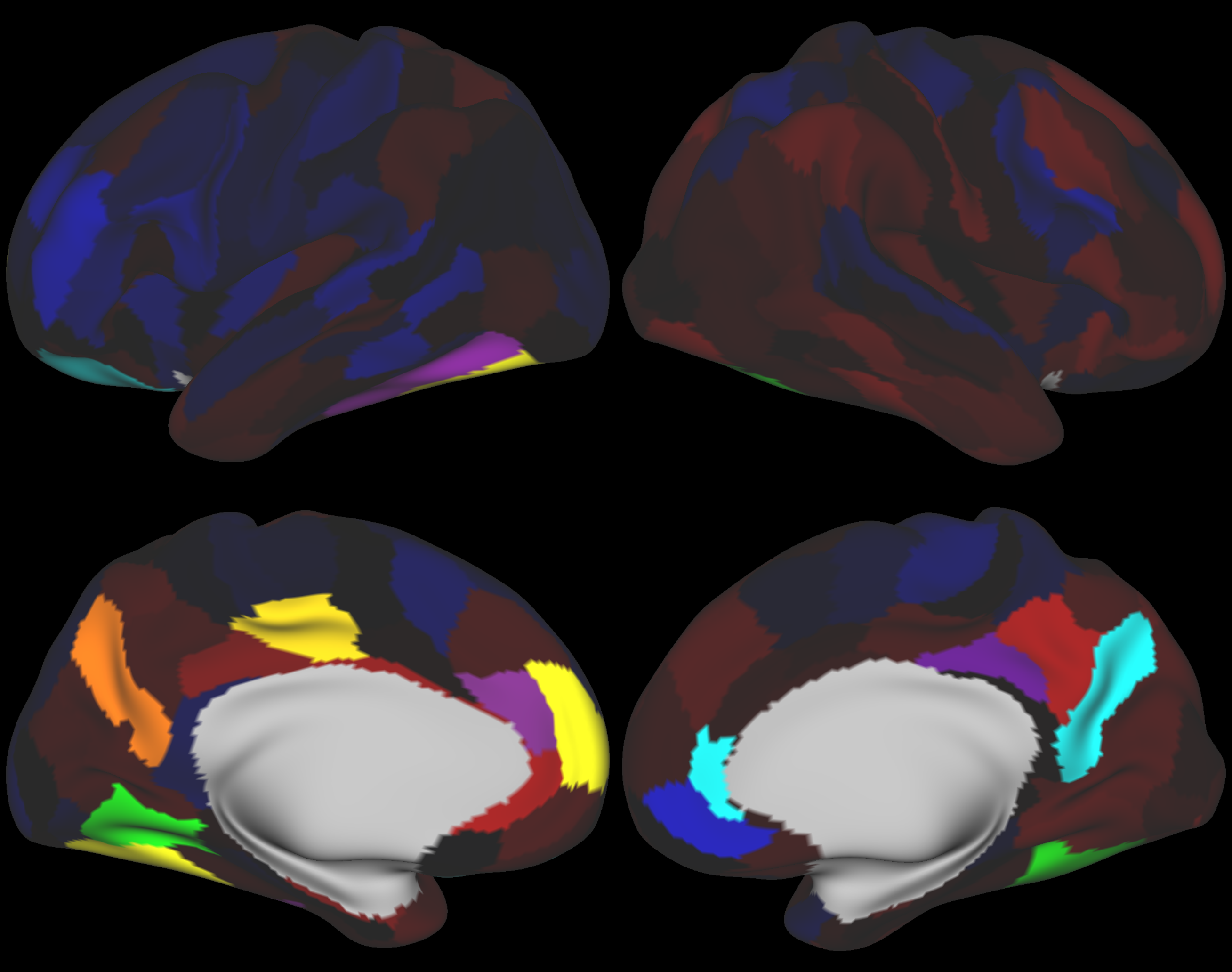}
    \subcaption{Cortical coefficients.}
    \label{fig:rde:prcca:coef:a}
     \end{subfigure}
  \hfill
  \begin{subfigure}[b]{0.5\textwidth}
    \includegraphics[width=\textwidth]{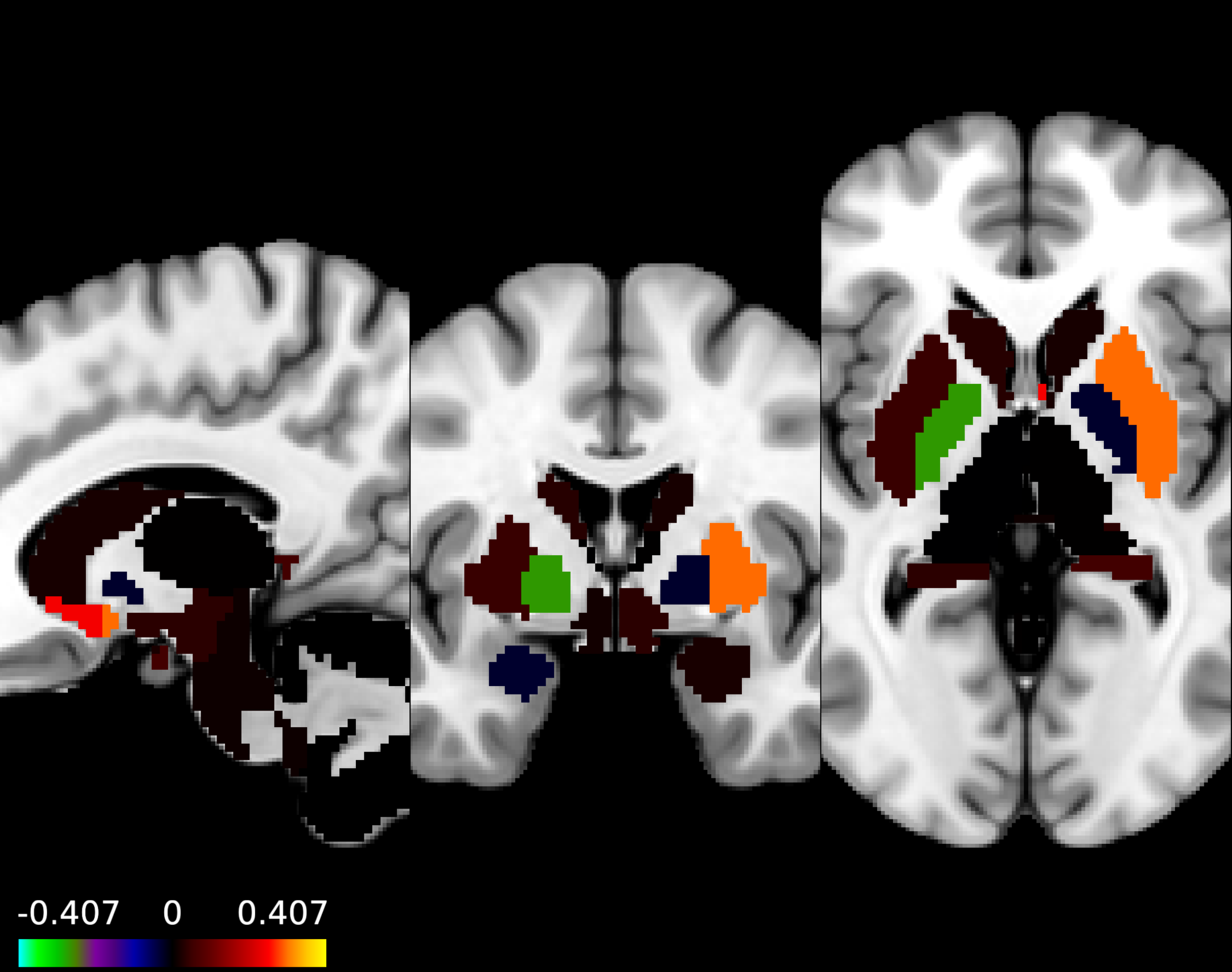}
    \subcaption{Subcortical coefficients.}
      \label{fig:rde:prcca:coef:b}
    \end{subfigure}
\caption{PRCCA coefficients computed for the averaged data ($229$ features) with hyperparameter $\lambda_1 = 1$ chosen by cross-validation.}
  \label{fig:rde:prcca:coef}
\end{figure}

\section{Canonical correlation analysis for grouped data}
\label{grcca}

\subsection{Handling data with a group structure}
\label{grcca:motivation}

The main critique of applying standard RCCA approach to the fMRI data is that, in fact, RCCA completely ignores the brain geometry treating all the features equally. 
Recall that in the Human Connectome data the features representing particular greyordinates are grouped into macro regions according to the function and anatomy. 
The goal of the \textit{group regularized canonical correlation analysis (GRCCA)} is to incorporate this underlying data structure into the regularization penalty.

There are some group extensions of CCA based on elastic net and group lasso penalizations suggested in the literature (see, for example, \citealp{group1} and \citealp{group2}). Unlike the existing methods the proposed GRCCA approach does not require an iterative algorithm and has a simple explicit solution. Equipped with the kernel trick it also allows working with data in a very high-dimensional feature space. 


GRCCA solves the CCA problem under the following two natural assumptions. First, we assume \textit{homogeneity of groups} and expect that the features within each group have approximately equal contribution to the canonical variates. In other words, the corresponding CCA coefficients do not vary significantly inside each group. Second, we assume the \textit{differentiating sparsity on a group level} and expect that the coefficients will be shrunk toward zero all together for some groups. In terms of brain imaging applications these two assumptions mean that greyordinates ``act in concert'' within each macro region and that some regions have a weaker effect on the studied phenomenon.

To state the GRCCA optimization problem rigorously we need to introduce some further notation. Suppose that the elements of random vectors $X$ and $Y$ are known to come in $K$ and $L$ groups, respectively. For each $k=1,\ldots,K$ let $p_k$ be the number of elements from $X$ that belong to group $k$, so $\sum_{k=1}^K p_k = p$, and let $X_k\in\mathbb{R}^{p_k}$ be the random vector that consists of these elements. Without loss of generality, we can assume that $X = \begin{psmallmatrix}
X_1\\
\ldots\\
X_K
\end{psmallmatrix}$. Similarly, one can divide the CCA coefficients $\alpha$ into blocks corresponding to different groups, i.e. $\alpha = \begin{psmallmatrix}
\alpha_1\\
\ldots\\
\alpha_K
\end{psmallmatrix}$. If $\bar\alpha_k = \frac{\mathbb{1}^\top\alpha_k}{p_k}$ denotes the mean of the CCA coefficients in group $k,$ then the group homogeneity assumption implies that all values of $\alpha_k$ do not deviate significantly from the mean value $\bar\alpha_k$, whereas, the differentiating sparsity on a group level implies that the average deviation of $\bar\alpha_k$ from zero is small. This can be characterized by the following two constraints:
$\sum_{k = 1}^K\|\alpha_k - \mathbb{1}\bar\alpha_k\|^2\leq t_1~$ and $~\sum_{k=1}^Kp_k\bar\alpha_k^2\leq s_1.$
Note that these two equations can be interpreted as bounds on within- and between- group variation,~respectively. One can derive similar constraints for $\beta$ coefficients as well.
Adding all the constraints to the CCA optimization problem we end up with the GRCCA optimization problem
\begin{align}
    \text{maximize} ~&\alpha^\top\widehat\Sigma_{XY}\beta~ \text{w.r.t.} ~\alpha\in\mathbb{R}^p~ \text{and} ~\beta\in\mathbb{R}^q \nonumber\\
    \text{subject to} ~\alpha^\top\widehat\Sigma_{XX}&\alpha = 1, ~\sum_{k = 1}^K\|\alpha_k - \mathbb{1}\bar\alpha_k\|^2\leq t_1, ~\sum_{k=1}^Kp_k\bar\alpha_k^2\leq s_1~ \nonumber\\
    \text{and} ~\beta^\top\widehat\Sigma_{YY}&\beta = 1, ~\sum_{\ell = 1}^L\|\beta_\ell - \mathbb{1}\bar\beta_\ell\|^2\leq t_2,~
    \sum_{\ell=1}^L q_\ell\bar\beta_\ell^2\leq s_2.
\label{sgrcca:op}
\end{align}

Next, denote $C_{m} = \frac{\mathbb{1}\mathbb{1}^\top}{m}\in\mathbb{R}^{m\times m}$. Let $C_X = C_{p_1}\oplus\ldots\oplus C_{p_K}\in\mathbb{R}^{p\times p}$
refer to the block diagonal matrix with blocks $C_{p_1},\ldots,C_{p_K}$.
Thus, the constraints on $\alpha$ can be rewritten in terms of a regularization penalty~as
\begin{align*}
&\alpha^\top\widehat\Sigma_{XX}\alpha + \lambda_1\sum_{k = 1}^K\|\alpha_k - \mathbb{1}\bar\alpha_k\|^2 +\mu_1\sum_{k=1}^K p_k\bar\alpha_k^2 =\nonumber\\
&\alpha^\top\widehat\Sigma_{XX}\alpha + \lambda_1\sum_{k = 1}^K\alpha_k^\top\left(I - C_{p_k}\right)\alpha_k +\mu_1\sum_{k=1}^K\alpha_k^\top C_{p_k}\alpha_k=\nonumber\\
&\alpha^\top\left(\widehat\Sigma_{XX} + \lambda_1(I-C_X) + \mu_1 C_X\right)\alpha=\alpha^\top\left(\widehat\Sigma_{XX}+K_X(\lambda_1,\mu_1)\right)\alpha.
\end{align*}
Here $K_X(\lambda_1,\mu_1) = \lambda_1(I-C_X)+\mu_1 C_X$ is the penalty matrix.
The constraints for $\beta$ can be combined in a similar way 
leading to the GRCCA modified correlation coefficient
\begin{align*}
    \rho_{GRCCA}(\alpha, \beta;\lambda_1, \mu_1, \lambda_2, \mu_2) = \frac{\alpha^\top\widehat\Sigma_{XY}\beta}{\sqrt{\alpha^\top(\widehat\Sigma_{XX}+ K_X(\lambda_1,\mu_1))\alpha}~\sqrt{\beta^\top(\widehat\Sigma_{YY}+\ K_Y(\lambda_2,\mu_2))\beta}}.
\end{align*}
Note that this correlation coefficient has similar structure to the RCCA coefficient (\ref{rcca:rho}) and PRCCA coefficient (\ref{prcca:rho}), but now the covariance matrices in the denominator are adjusted by block diagonal matrices $K_X(\lambda_1,\mu_1)$ and $K_Y(\lambda_2,\mu_2)$. 

Similar to RCCA and PRCCA, the explicit solution to the GRCCA problem can be found via the SVD of matrix 
$\left(\widehat\Sigma_{XX} + K_X(\lambda_1,\mu_1) \right)^{-\frac12}\widehat\Sigma_{XY}\left(\widehat\Sigma_{YY} + K_Y(\lambda_2,\mu_2) \right)^{-\frac12},$ which can be problematic in high dimensions. It turns out that there is a simple linear transformation that converts the GRCCA problem to an equivalent RCCA/PRCCA problem. In Supplement Sections \ref{supp:grcca:solution} and \ref{supp:grcca:solution:alternative} we give two ways: via the SVD of the penalty matrix and via feature matrix extension. This link can be subsequently used to establish the kernel trick for group-structured data thereby reducing computations in high dimensions.

\subsection{Link to the flexibility vs. performance trade-off}
\label{grcca:interpretability}
There are several important properties of the proposed penalty matrix explaining the motivation for the GRCCA method. First, one can show that for $K_X(\lambda_1,\lambda_1) = \lambda_1I$, so RCCA is the special case of GRCCA. Second, increasing $\lambda_1$ restrains the variability of coefficients within each brain region and, in limit, makes all the coefficients that belong to the same brain region equal to each other (and equal to the region mean). This is essentially equivalent to replacing features in each brain region by the average.
Therefore, when $\lambda_1\to\infty$ the GRCCA problem becomes equivalent to the RCCA problem solved for the reduced data (see Section \ref{connectome:rcca} for the details). 
To sum up, varying $\lambda_1$ and $\mu_1$ allows us to approach the RCCA method conducted for either full or reduced data thereby controlling the flexibility vs interpretability trade-off described in Section~\ref{rcca:interpretability}.

\subsection{GRCCA for the Human Connectome study}
\label{grcca:connectome}

In this section we apply the GRCCA method to the Human Connectome study data grouping activation features according to the brain regions. We again adjusted $X$ and $Y$ for the sex effect and ran 10-fold cross-validation on the adjusted data pair with the penalty factors varying in the grid  $\lambda_1, \mu_1 = 10^{-3}, 10^{-2}, \ldots, 10^4, 10^5$ (see Section \ref{cv} for the details). The resulting cross-validation curves are presented in Figure \ref{fig:rde:scores}. According to the plot, the highest cross-validation score is attained for $\lambda_1 = 100$ and $\mu_1 =~1$ leading to the correlation of $0.296$.

\begin{figure}[h!]
  \centering
    \includegraphics[width=0.96\linewidth]{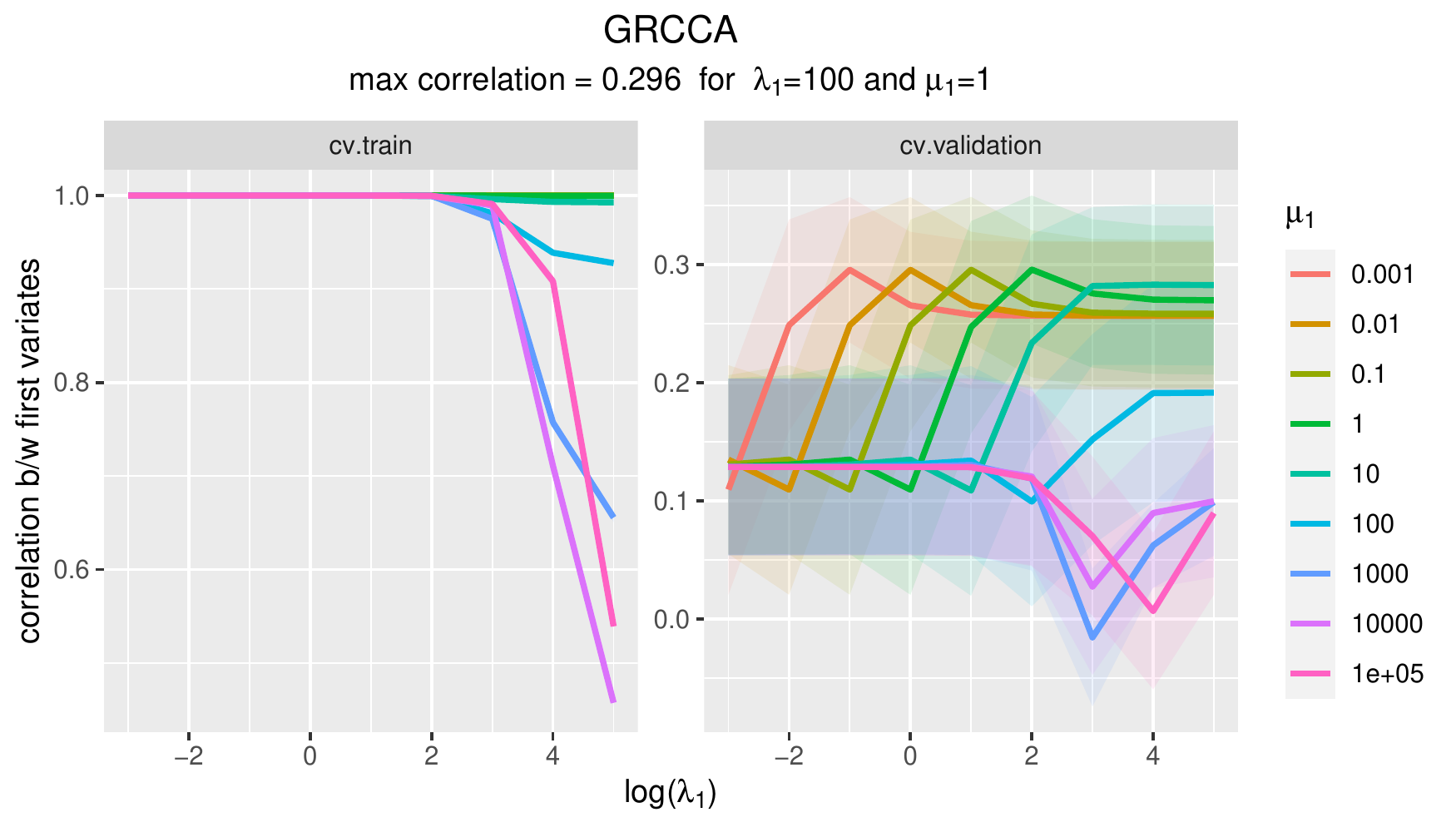}
 \caption{The cross-validation curves obtained via GRCCA with 10-fold cross-validation for the Human Connectome Project dataset. Left panel: (unpenalized) correlation between train canonical variates. Right panel: (unpenalized) correlation between validation canonical variates.}
  \label{fig:rde:scores}
\end{figure}

\begin{figure}[h!]
  \centering
    \includegraphics[width=0.9\linewidth]{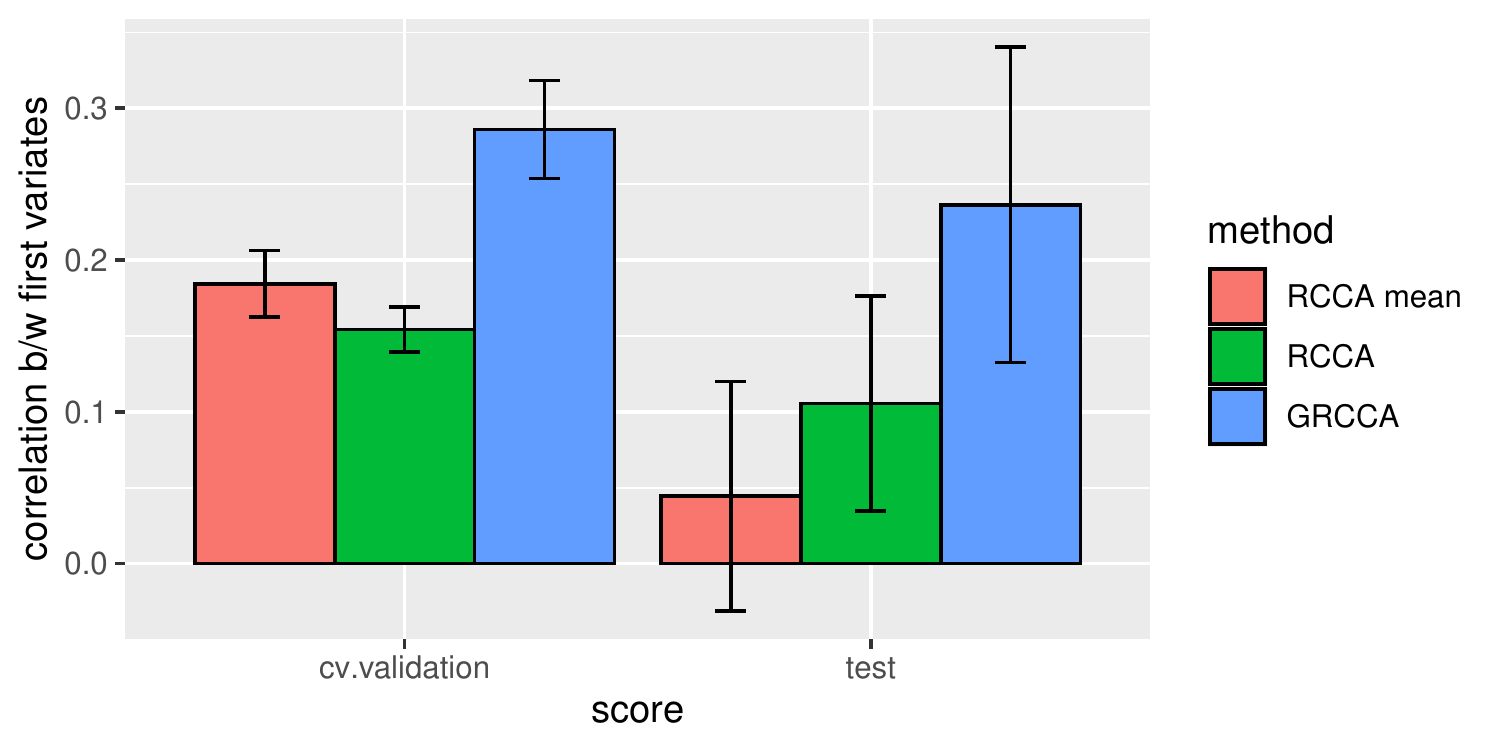}
  \caption{Nested cross-validation scores computed for three models: RCCA fitted on the averaged data (229 features), RCCA fitted on the original data ($90~368$ features), GRCCA. Two scores are reported: \textit{cv.validation} = maximum score obtained via 10-fold cross-validation, averaged across 11 NCV folds; \textit{test} = score computed on independent test set, averaged across 11 NCV folds.}
  \label{fig:rcca:ncv}
\end{figure}

 To validate these significant findings, we again run nested cross-validation (see Section \ref{connectome:rcca}). According to Figure \ref{fig:rcca:ncv}, the NCV average test set score is equal to $0.236$. Thus, although slightly optimistic (by a modest value of $0.06$), the GRCCA cross-validation correlation is not a spurious finding and is a significant improvement comparing to the RCCA method.

In addition to better performance, the GRCCA method allows us to track the effect of the variation inside each brain region (controlled by the $\lambda_1$ penalty factor) on the resulting canonical correlation separately from the effect of the variation across the regions (controlled by $\mu_1$ penalty factor). For example, the spikes for small $\mu_1$ values suggest that it is more beneficial to reduce within group variation than between group one. In other words, shrinking all brain region coefficients toward the group means improves the performance more than shrinking group means toward zero. Moreover, for small $\mu_1$ it is the ratio of hyperparameters that plays the key role: the highest score is always achieved when
$\frac{\lambda_1}{\mu_1} = 100.$ For large $\mu_1$ this pattern disappears as we start to over-penalize both between and within group variations.

\subsection{Using GRCCA for visualization}
\label{grcca:vizualization}

In this section we demonstrate another advantage of GRCCA in the context of visualization and interpretability.
In Figure \ref{fig:rde:path} we present the coefficient paths ($\alpha$ vs. $\lambda_1$) produced by the RCCA method as well as the group modification. Here different colors represent different brain regions. According to the plot, we observe the following behavior of the coefficients. For the RCCA method the canonical coefficients are shrunk toward zero all together with the growth of $\lambda_1$. On the contrary, for large $\lambda_1$ all the GRCCA coefficient paths become horizontal, which implies the convergence of the coefficients to the group means. Finally, increasing $\mu_1$ shrinks all the group means toward zero encouraging differentiating sparsity on a group level.

\begin{figure}[htb!]
  \centering
  \begin{subfigure}{0.5\textwidth}
    \includegraphics[width = \textwidth]{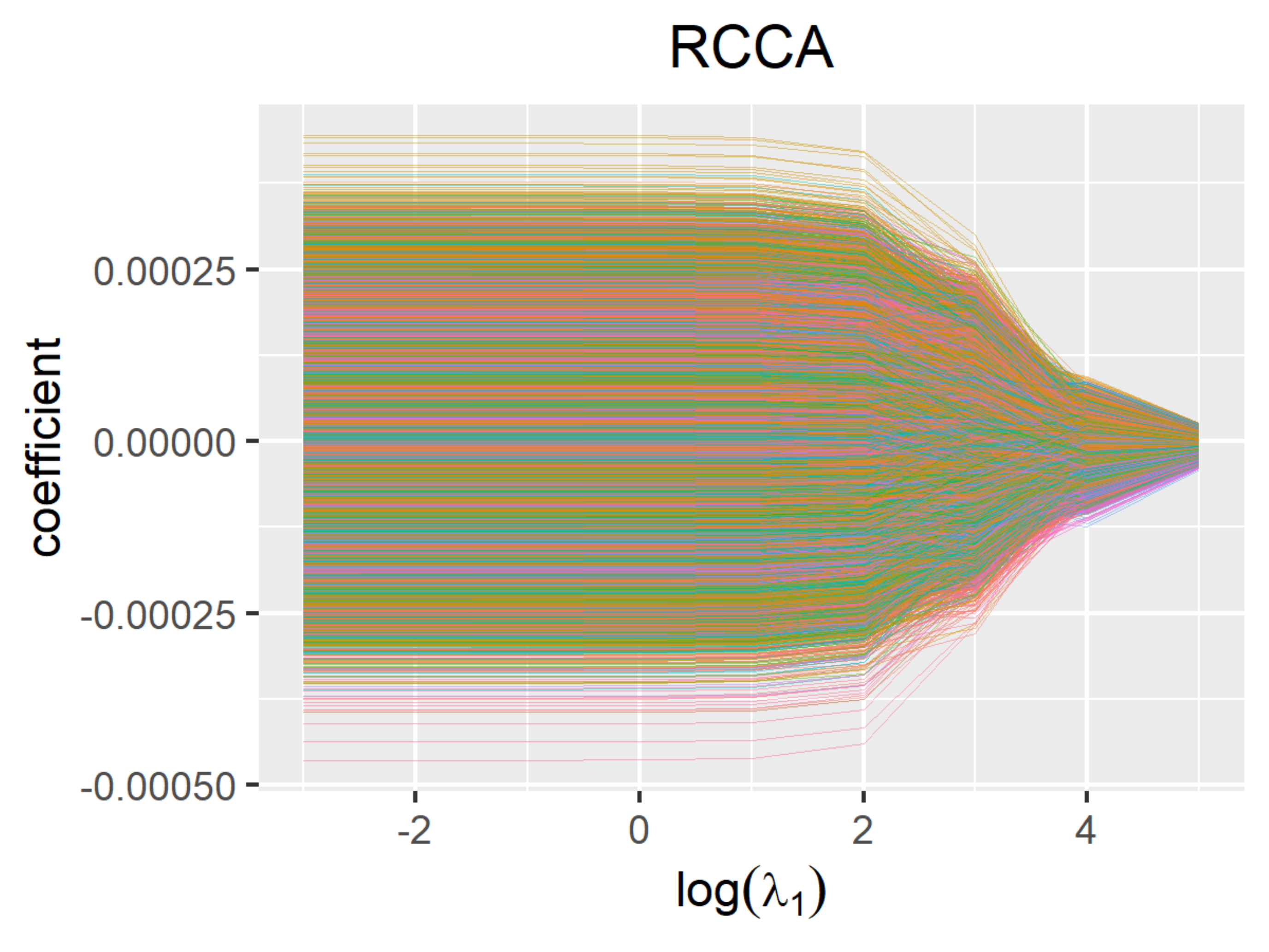}
    \end{subfigure}
    \vfill
     \begin{subfigure}{\textwidth}
    \includegraphics[width = \textwidth]{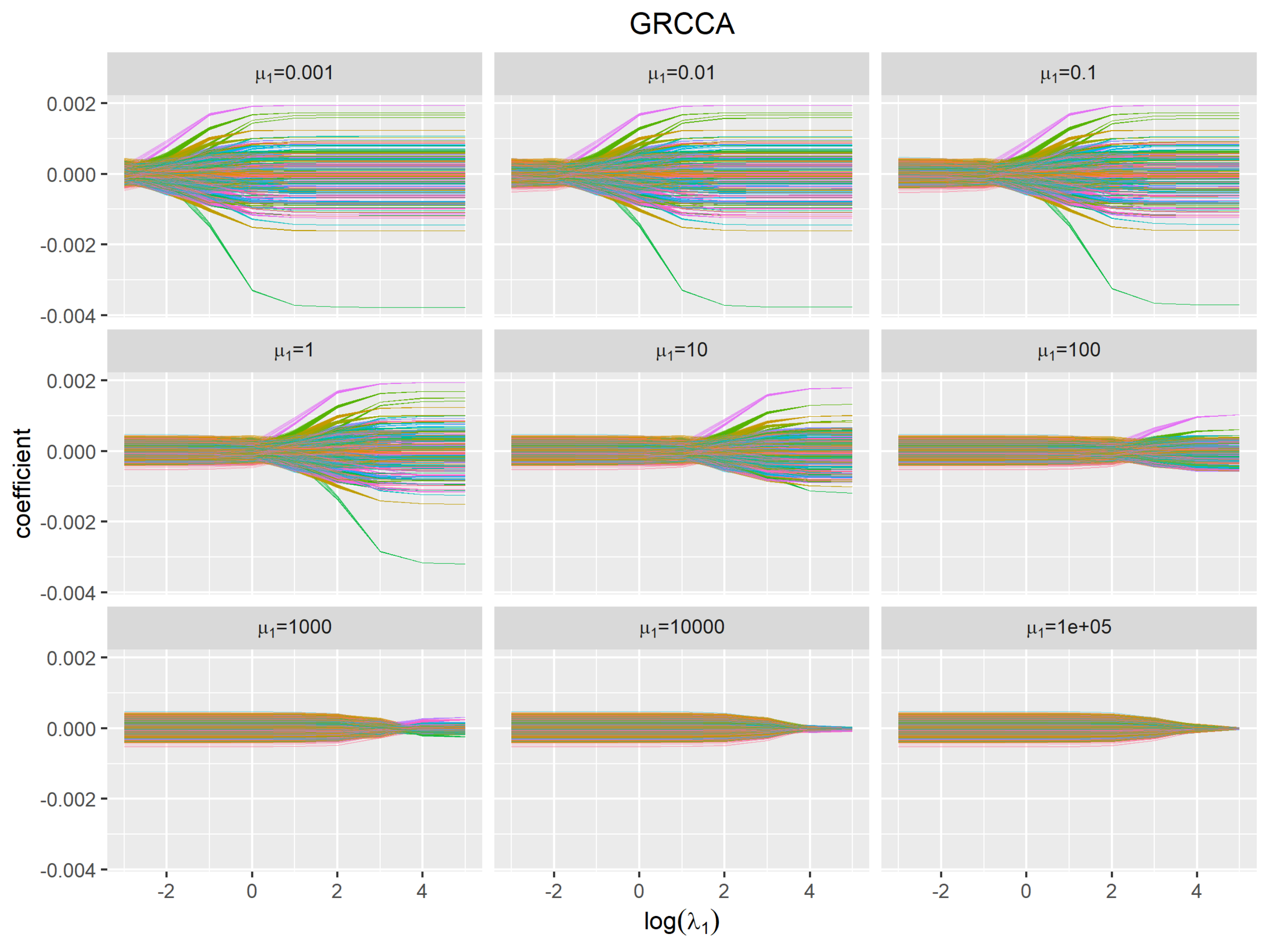}
     \end{subfigure}
 \caption{Coefficient paths for RCCA and GRCCA with colors representing the regions. The RCCA method shrinks canonical coefficients toward zero all together with the growth of $\lambda_1$. The GRCCA method shrinks them toward the group means with the growth of $\lambda_1$, whereas increasing $\mu_1$ shrinks all the group means toward zero.}
   \label{fig:rde:path}
\end{figure}

In Figure \ref{fig:rde:GRCCA:coef} we present the brain images computed for $\mu_1 = 1$ and $\lambda_1 = 1, 10, 100, 1000$. 
Note that larger $\lambda_1$ makes the brain region pattern more obvious (similar to Figures \ref{fig:rde:rcca:coef:a}-\ref{fig:rde:rcca:coef:b}). Moreover, for $\lambda_1 = 100$ we get the plot corresponding to the best GRCCA model from Section~\ref{grcca:connectome}. Thus the GRCCA model chosen by cross-validation has not only better performance on the validation set than the RCCA model, but it is also more interpretable in the context of the importance of each brain region.
Specifically, the canonical component had especially high positive loadings in subcortical regions involved in reward processing, such as the striatum (nucleus accumbens, putamen) and thalamus \citep{haberChapterAnatomyConnectivity2017}. It also loaded positively on a cortical network encompassing the temporal lobe, dorsolateral prefrontal, dorsomedial prefrontal, posterior cingulate and precentral cortices (see Figure \ref{fig:rde:GRCCA:coef} for an annotated visualization of these results). Most of these regions have been shown to be connected to the striatum and to be part of key reward-processing pathways as well \citep{haberChapterAnatomyConnectivity2017}.

\begin{figure}[!htb] 
    \begin{subfigure}[b]{0.37\textwidth}
    \includegraphics[width=\textwidth]{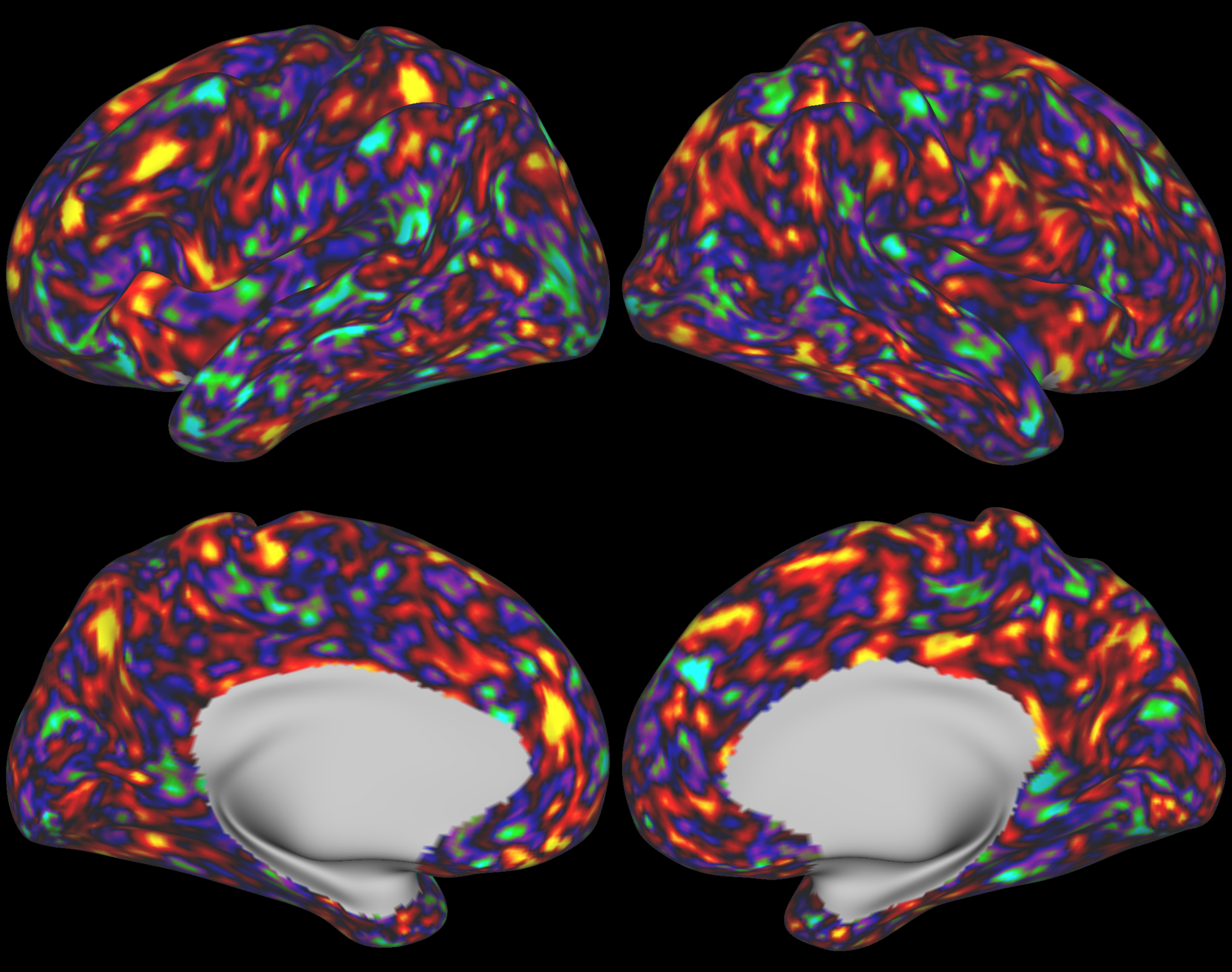}
     \end{subfigure}
  \hfill
  \begin{subfigure}[b]{0.37\textwidth}
    \includegraphics[width=\textwidth]{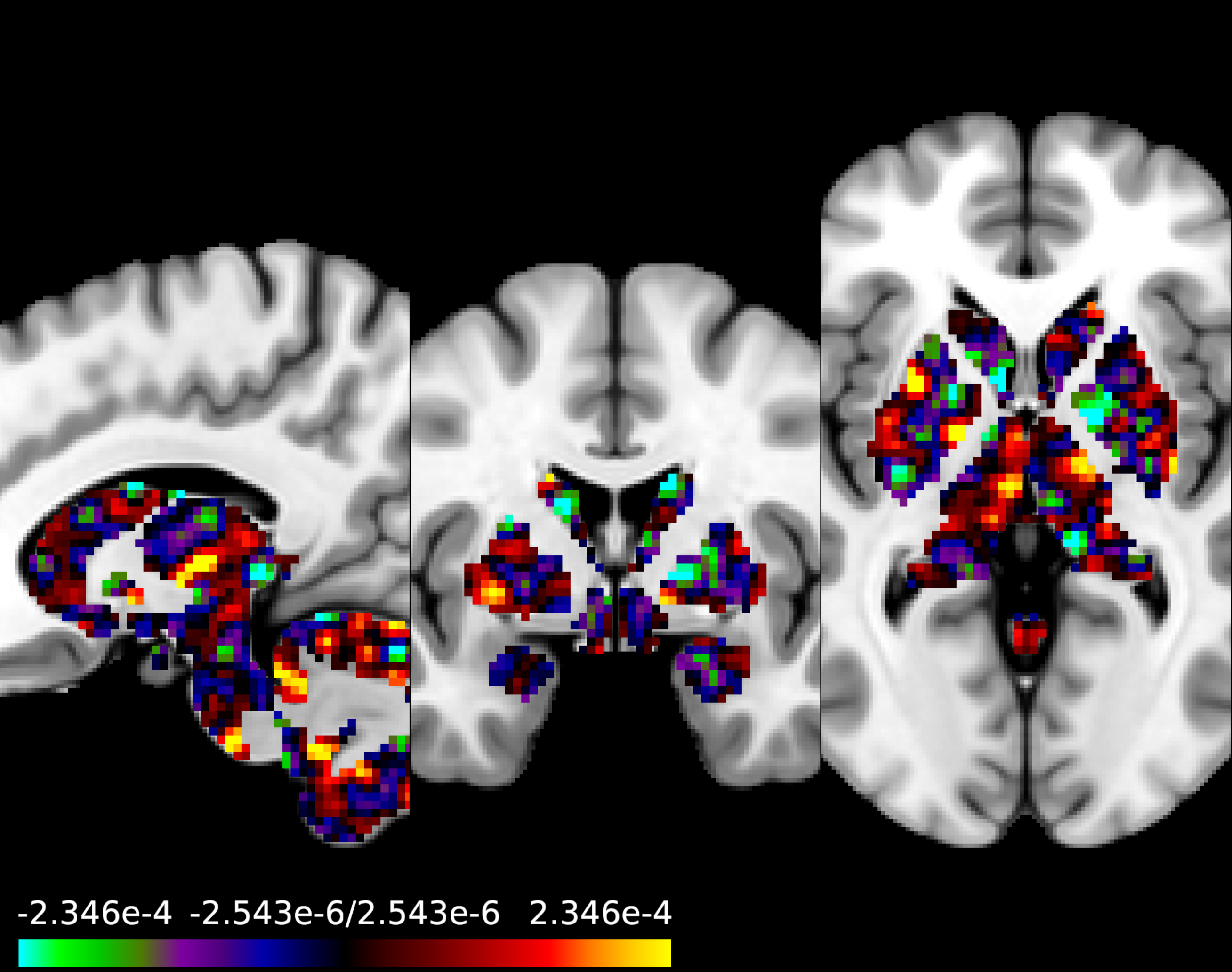}
    \end{subfigure}
    \vfill
 \begin{subfigure}[b]{0.37\textwidth}
    \includegraphics[width=\textwidth]{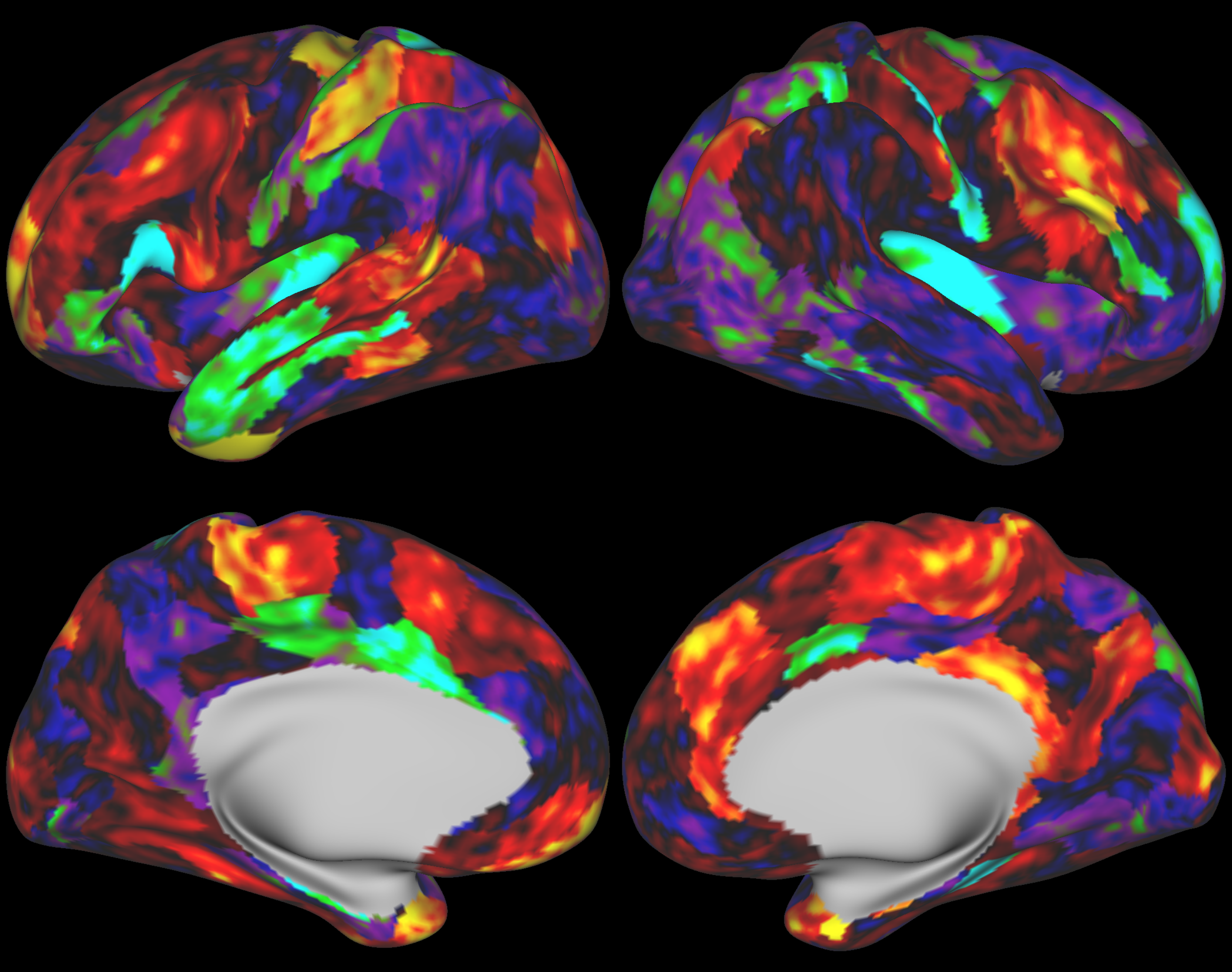}
     \end{subfigure}
  \hfill
  \begin{subfigure}[b]{0.37\textwidth}
    \includegraphics[width=\textwidth]{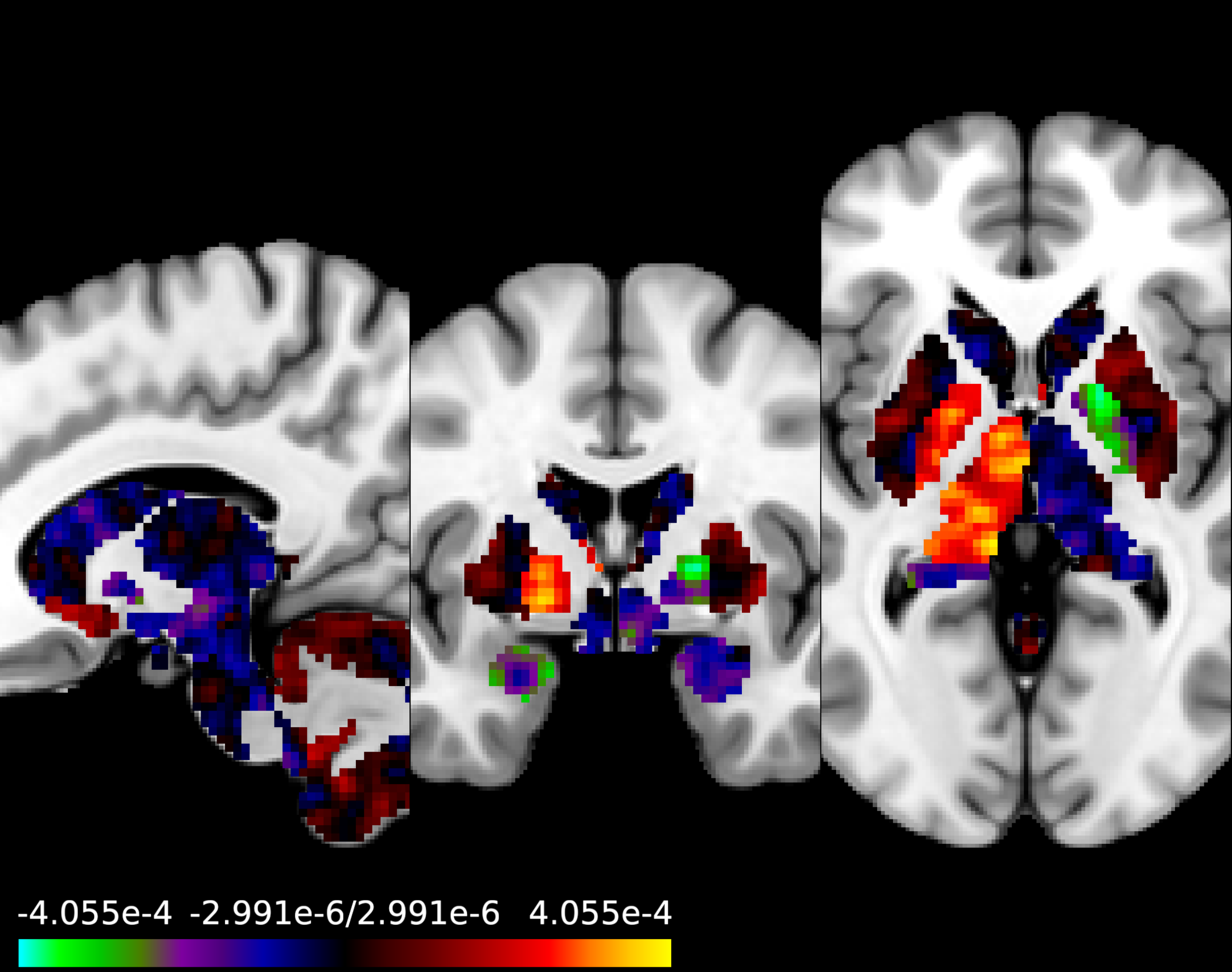}
    \end{subfigure}
    \vfill
 \begin{subfigure}[b]{0.37\textwidth}
    \includegraphics[width=\textwidth]{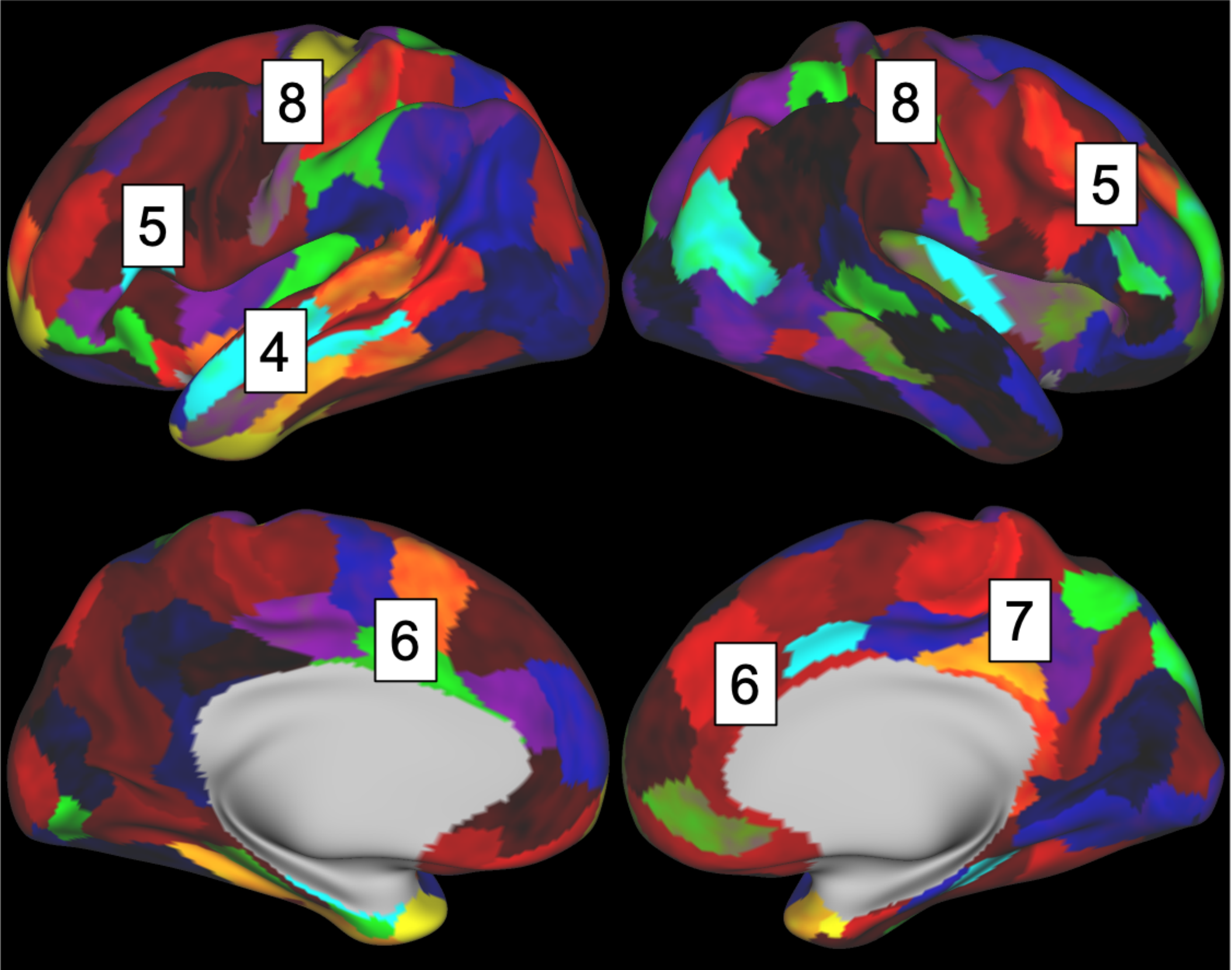}
     \end{subfigure}
  \hfill
  \begin{subfigure}[b]{0.37\textwidth}
    \includegraphics[width=\textwidth]{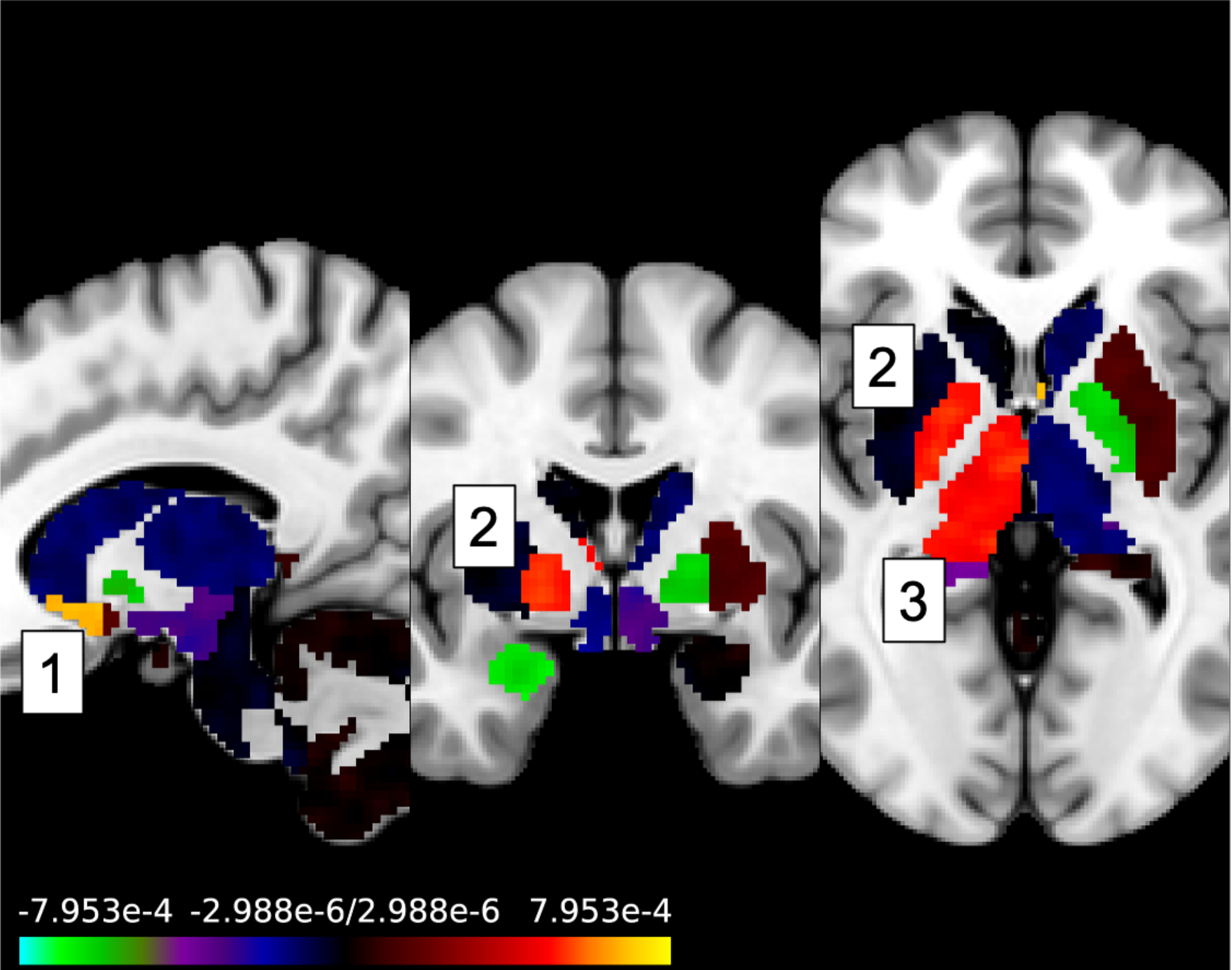}
    \end{subfigure}
    \vfill
 \begin{subfigure}[b]{0.37\textwidth}
    \includegraphics[width=\textwidth]{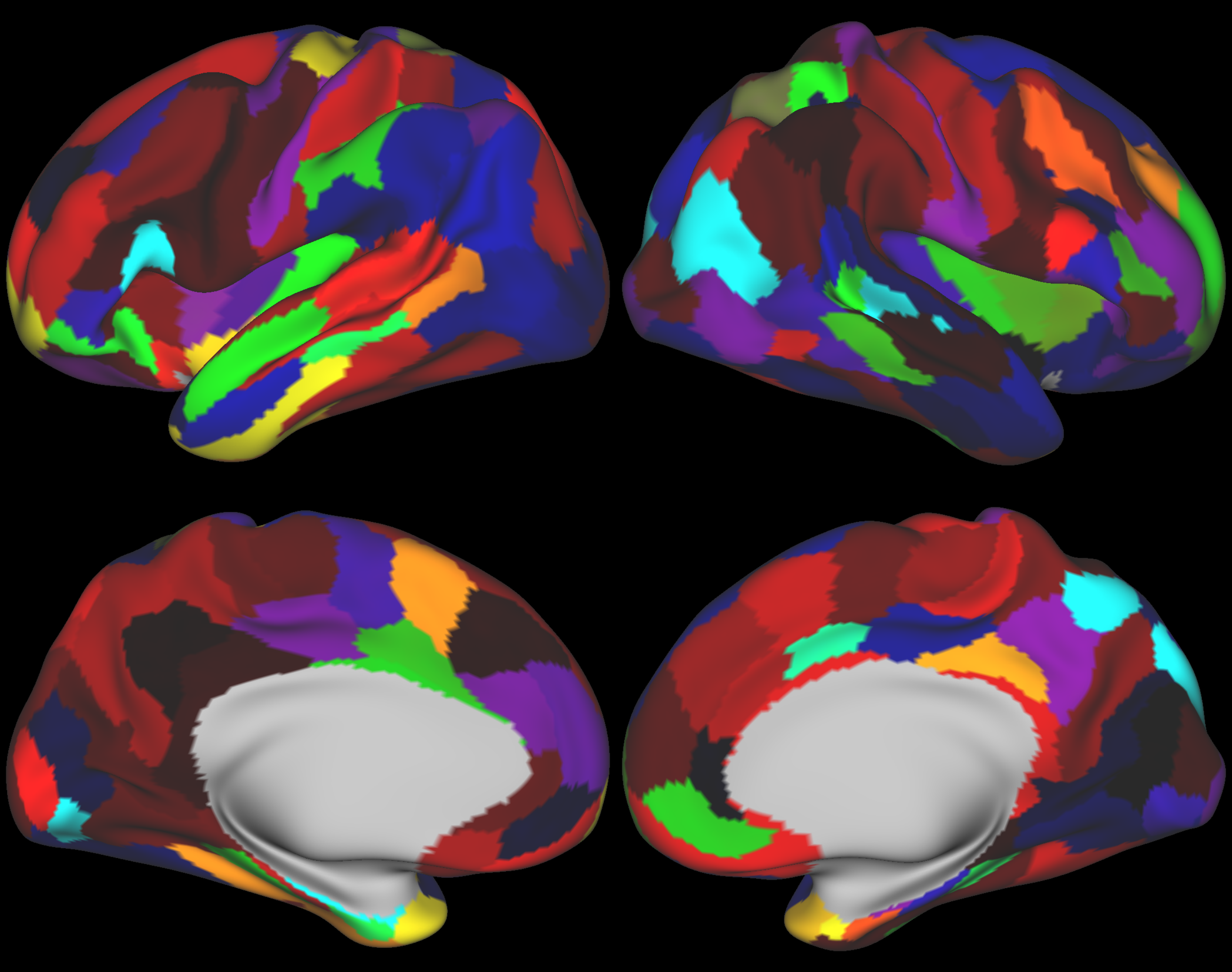}
     \end{subfigure}
  \hfill
  \begin{subfigure}[b]{0.37\textwidth}
    \includegraphics[width=\textwidth]{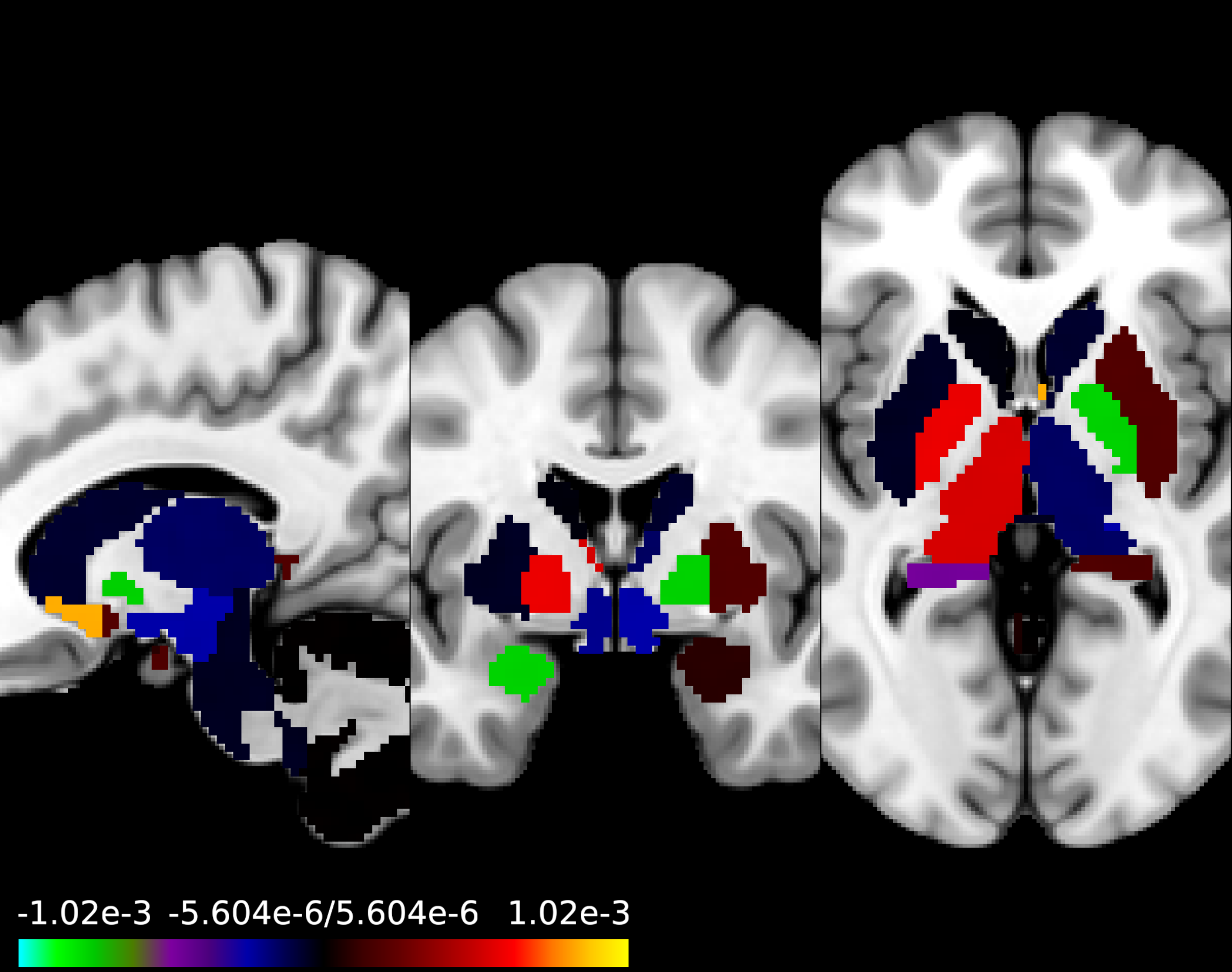}
    \end{subfigure}

\caption{From top to bottom: GRCCA coefficients for $\mu_1 = 1$ and $\lambda_1 = 1, 10, 100, 1000$. The third row represents the solution produced by the cross-validation procedure. Annotation of brain regions: [1] nucleus accumbens, [2] putamen, [3] thalamus, [4] temporal lobe, [5] dorsolateral prefrontal cortex, [6] dorsomedial prefrontal cortex, [7] posterior cingulate cortex, [8] precentral cortex.} 
  \label{fig:rde:GRCCA:coef}
\end{figure}

\section{General approach to regularization}
\label{genrcca}

It turns out that all RCCA, PRCCA and GRCCA methods are similar in nature: they perform regularization by means of adjusting covariance matrices $\widehat\Sigma_{XX}$ and/or $\widehat\Sigma_{YY}$ in the denominator of the modified correlation coefficient. In this section we consider the class of CCA problems with general weighted $\ell_2$ regularization. 

If $K_X, K_Y\in\mathbb{R}^{p \times p}$ are some positive semi-definite \textit{penalty matrices}, then the general modified correlation coefficient can be written as  
\begin{align}
    \rho(\alpha, \beta; K_X, K_Y) = \frac{\alpha^\top\widehat\Sigma_{XY}\beta}{\sqrt{\alpha^\top(\widehat\Sigma_{XX} + K_X)\alpha}~\sqrt{\beta^\top(\widehat\Sigma_{YY}+K_Y)\beta}}.
\label{genrcca:rho}
\end{align}

The accompanying \textit{general RCCA} optimization problem is therefore
\begin{align}
    &\text{maximize} ~\alpha^\top\widehat\Sigma_{XY}\beta~ \text{w.r.t.} ~\alpha\in\mathbb{R}^p~ \text{and} ~\beta\in\mathbb{R}^q \nonumber\\
    \text{subject to} ~\alpha^\top&\widehat\Sigma_{XX}\alpha = 1, ~\alpha^\top K_{X}\alpha\leq t_1~
    \text{and} ~\beta^\top\widehat\Sigma_{YY}\beta = 1, ~\beta^\top K_{Y}\beta\leq t_2.
\label{genrcca:op}
\end{align}

Note that the inequality constraints in (\ref{genrcca:op}) can be rewritten as $\|\alpha\|_{K_X}\leq t_1$ and $\|\beta\|_{K_Y}\leq t_2$, where $\|\cdot\|_A$ is weighted Euclidean norm defined as $\|x\|_A = x^\top Ax.$ The resulting canonical variates and coefficients can be found via the singular value decomposition of the matrix $\left(\widehat\Sigma_{XX} + K_X \right)^{-\frac12}\widehat\Sigma_{XY}\left(\widehat\Sigma_{YY} + K_Y \right)^{-\frac12}$. 
To handle General RCCA in high dimensions, one can link it to the two methods for which we already established the kernel trick. In the Supplement Section \ref{supp:genrcca:solution} we provide the proof of the following lemma.

\begin{lemma*}[General RCCA to RCCA/PRCCA]
If both $K_X$ and $K_Y$ are positive definite then, by some proper change of basis, the general RCCA problem can be reduced to the RCCA one. Alternatively, if one of $K_X$ and $K_Y$ has zero eigenvalues then general RCCA boils down to solving the PRCCA problem with number of unpenalized coefficients equal to the multiplicity of the zero eigenvalue.
\end{lemma*}

\section{Simulation study}
\label{simulation}

\subsection{Generating data with a group structure}
\label{simulation:data}

In this section we set up a small simulation experiment where we compare performance of all the above methods on the data with group structure. We generate the data as follows. For random vector $X$ we assume that it is grouped into $K$ groups of equal size, thereby having $p_k = \frac p K$ variables in group $k$. Each group of $X$ is generated by one of $K$ centroid random variables and is obtained by adding some Gaussian noise to the centroid. Moreover, we assume the presence of some correlation between the centroids and $Y$.
To be precise, to generate the data we exploit the multivariate normal distribution as a joint distribution of random vector $Y\in\mathbb{R}^q$ and random vector of centroids $X^c\in\mathbb{R}^K$: 
\begin{center}
$(Y, X^c)\sim\mathcal{N}_{q+K}(0, \Sigma)$ ~~with~~ $\Sigma = \begin{psmallmatrix}
I_{q} & \mathbb{1}\mathbb{1}^\top \sigma_{XY}^2\\
\mathbb{1}\mathbb{1}^\top \sigma_{XY}^2 & I_{K}
\end{psmallmatrix}.$
\end{center}
Next, we generate random vector $X_k\in\mathbb{R}^{p_k}$ corresponding to groups $k = 1,\ldots, K$ from the distribution
$X_k|X^c_k\sim\mathcal{N}_{p_k}(\mathbb{1}X^c_k,\sigma_X^2I),$
where $X_k^c$ is the $k$-th component of the centroid vector. Finally, we obtain matrices $\mathbf{X} = (\mathbf{X}_1, \ldots, \mathbf{X}_K)\in\mathbb{R}^{n\times p}$ and $\mathbf{Y}\in\mathbb{R}^{n\times q}$ by drawing $n$ samples from the above distributions. In our experiments we use $n = 10$, $p = 15$ and $q = 3$, the number of groups is $K = 5$. We set $\sigma_{X} = 1$ and test two settings: $\sigma_{XY} = 0.5$ for correlated data and
 $\sigma_{XY} = 0$ for independent data.

As the next step, we run RCCA, PRCCA and GRCCA on the generated data imposing the regularization on the $X$ part only and using the following hyperparameters.
For all methods the penalty factor is chosen to be $\lambda_1 = 10^{-5}, 10^{-4}, \ldots, 10^4, 10^5$.
For PRCCA we penalize $p_1 = 10$ variables only leaving $p_2 = 5$ variables untouched; these five unpenalized variables correspond to the first features in each group. For GRCCA we again try two versions. First we run GRCCA with $\mu_1 = 0$, i.e. without differentiating sparsity on a group level assumption. Next we add differentiating sparsity to GRCCA and vary the penalty factor in the range $\mu_1 = 10^{-4}, 10^{-3}, \ldots, 1, 10$.

We compare all methods in terms of resulting correlations. For this purpose, we generate $1000$ train and test sets, fit models on train and evaluate canonical correlation value on test. We plot average train and test correlations as well as their one standard error intervals vs. penalty factor $\lambda_1$ (see Fig. \ref{fig:traintest}). According to the plot, for the correlated data the best test score is achieved by sparse version of GRCCA, which significantly outperforms RCCA. Better performance can be explained by the presence of the groups structure in the data. Note that non-sparse GRCCA also looses in terms of the test score. The possible reason is that number of observations ($n = 10$) is only twice as large as the number of groups ($K = 5$), so regularization on a group level helps to prevent overfitting to train data. In the case of independent data, all the competitors perform in a similar way: the average test correlation is very close to zero regardless the hyperparameter value; the test correlation curves are almost flat.    

\begin{figure}[h!]
\centering
\includegraphics[width = \textwidth]{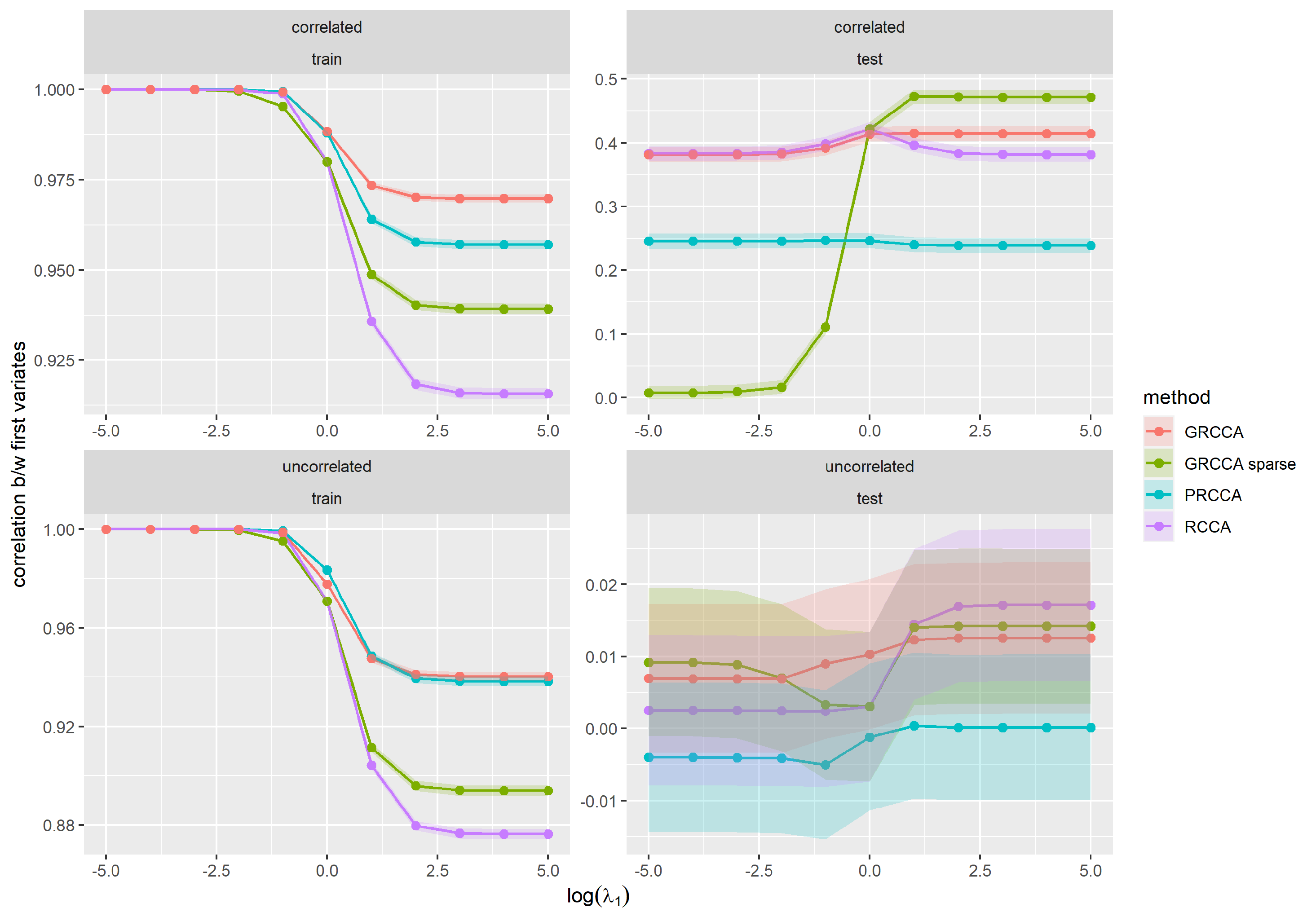}
\caption{Train and test curves computed via simulation. Four models presented: RCCA, PRCCA and GRCCA with zero ($\mu_1 = 0$) and non-zero sparsity ($\mu_1 = 1$). First row: train and test correlation obtained for data with correlation ($\sigma_{XY} = 0.5$). Second row: train and test correlation obtained for uncorrelated data ($\sigma_{XY} = 0$).}
\label{fig:traintest}
\end{figure}

\begin{figure}[h!]
  \centering
    \includegraphics[width=0.9\linewidth]{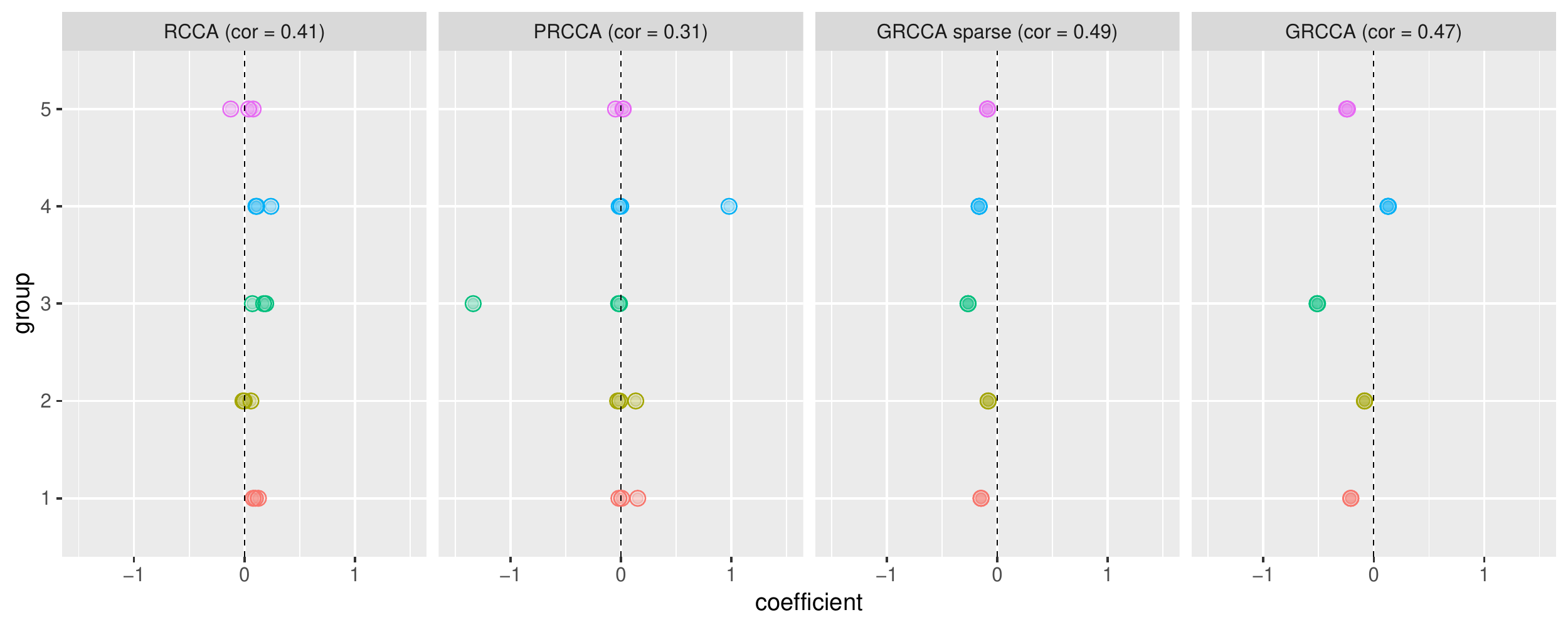}
  \caption{Coefficient values obtained via RCCA, PRCCA and two versions of GRCCA with hyperparameter~$\lambda_1$ chosen to maximize the test correlation. In this plot color corresponds to the group number.}
  \label{fig:coefs}
\end{figure}

Finally, for each model we pick the value of $\lambda_1$ according to the maximum test score and compare the CCA coefficients $\alpha$ for the chosen models. Figure \ref{fig:coefs} displays the main difference between the RCCA and GRCCA methods.
Although both techniques aim to reduce the data dimensionality, the reduction is achieved in a different way. Unlike RCCA, which treats all the coefficients equally and confines their deviation from zero, GRCCA carries out the reduction treating equally the coefficients inside each group and removing the within group noise. To sum up, in the presence of a group structure the group modification of RCCA allows for dimensionality reduction in a more efficient and interpretable way.

\section{Discussion}
\label{discussion}
In this paper we proposed several approaches to the CCA regularization. The introduced PRCCA technique has a similar flavor as RCCA, but it penalizes only a subset of canonical coefficients. Both of these methods combined with the proposed kernel trick allows us to find the CCA solution even in case of extremely high dimensional data.
We further present the GRCCA method, which is based on the underlying group structure of the data, and extend regularization to the case of a more general regularization penalty thereby proposing General RCCA. The close connection between the latter techniques with RCCA and PRCCA methods enables to utilize the kernel trick in the general case thus providing a powerful tool for regularizing CCA in the high-dimensional framework.    

There is still much scope for future work. One interesting direction for further research is to consider other applications of the proposed group RCCA technique. For example, there are many problems in genetics where genes are grouped by functional similarity. Further, in this paper we cover only two types of penalties: partial and group; although the proposed kernel trick can handle any $\ell_2$-type penalty (see Section \ref{genrcca} for general RCCA). Thus we can study other structured modifications of RCCA that can be beneficial for applications. As an example, it may be interesting to explore hierarchical group structure, where not only brain loci are combined in some regions, but also regions are combined in some groups (e.g. we have cortical and subcortical groups in the HCP study).  

From the computational point of view, it would be useful to investigate how one can optimize the choice of the hyperparameters. The following idea is inspired by ridge regression, which also uses the $\ell_2$ penalty. 
Note that currently it is not necessary to apply any data normalization before running regularized CCA (all the fMRI features have the same scale). However, the overall scale of $X$ influences 
the choice of the hyperparameters, i.e. multiplying $X$ by some number $a$ implies increasing the penalty factors by $a^2$ times as well. Therefore, it would be beneficial to develop some recommendations for the grid of hyperparameters the user should search through. For instance, we can use the ridge regression heuristic and introduce a concept of degrees-of-freedom (aka $df$) for CCA with regularization. Then we can base the hyperparameter recommendation on the $df$ value. 

\section{Software}

Proposed methods are implemented in the R package {\tt RCCA}; the software is available from Github
(\url{https://github.com/ElenaTuzhilina/RCCA}).

\section{Funding}
 Leonardo Tozzi was supported by grant U01MH109985 under PAR-14-281 from the National Institutes of Health. Trevor Hastie was partially supported by grants DMS-2013736 and IIS
1837931 from the National Science Foundation, and grant 5R01 EB
001988-21 from the National Institutes of Health.

\begin{center}
\section*{\huge{Supplementary materials}}
\end{center}

\section{Proof of RCCA kernel trick lemma}
\label{supp:rcca:kernel}

\begin{lemma*}[RCCA kernel trick]
The original RCCA problem stated for $\mathbf{X}$ and $\mathbf{Y}$ can be reduced to solving the RCCA problem for $\mathbf{R}$ and $\mathbf{Y}$. The resulting canonical correlations and variates for these two problems coincide. The canonical coefficients for the original problem can be recovered via the linear transformation $\alpha_X = V\alpha_R.$
\end{lemma*}

\begin{proof}
Denote $V^\perp$ an orthogonal complement of matrix $V$, i.e. the matrix $V^\perp\in~\mathbb{R}^{p \times (p-n)}$ such that $\widetilde V = ~\left(V,V^\perp\right)\in ~\mathbb{R}^{p \times p}$ is a full-rank orthogonal matrix. Then 
$V^\top \widetilde V =~\left(I_n, 0\right).$ 
Denote also 
$$\gamma = \widetilde V^\top\alpha = 
\begin{psmallmatrix}
V^\top \alpha\\
\left(V^\perp\right)^\top\alpha
\end{psmallmatrix}=
\begin{psmallmatrix}
\gamma_1\\
\gamma_2
\end{psmallmatrix}.$$
Note that since there is a one-to-one correspondence betwineen $\alpha$ and $\gamma,$ the optimization of $\rho_{RCCA}(\alpha, \beta; \lambda_1, \lambda_2)$ w.r.t. to $\alpha$ is equivalent to optimization w.r.t. $\gamma$.
Further, the following relation is true
\begin{align*}
&\alpha^\top \frac {\mathbf{X}^\top\mathbf{Y}}{n}\beta = \gamma_1^\top\frac{\mathbf{R}^\top\mathbf{Y}}{n}\beta\\
& \alpha^\top\left(\frac{\mathbf{X}^\top\mathbf{X}}{n}+\lambda_1 I\right)\alpha = \gamma_1^\top\frac{\mathbf{R}^\top\mathbf{R}}{n}\gamma_1 + \lambda_1\gamma_1^\top\gamma_1 + \lambda_1\gamma_2^\top\gamma_2
\end{align*}
Therefore, the correlation coefficient (\ref{rcca:rho}) with $\lambda_2 = 0$ can be rewritten in terms of $R$ and $Y$ as
$$\rho_{RCCA}(\gamma, \beta;\lambda_1) = \dfrac{\gamma_1^\top\widehat\Sigma_{RY}\beta}{\sqrt{\gamma_1^\top(\widehat\Sigma_{RR}+\lambda_1 I)\gamma_1 + \lambda_1\|\gamma_2\|^2}\sqrt{\beta^\top\widehat\Sigma_{YY}\beta}}.$$
It is easy to show that the maximum value of $\rho_{RCCA}(\gamma, \beta;\lambda_1)$ is attained when $\gamma_2 = 0,$ so the above correlation coefficient is nothing but the RCCA correlation coefficient computed for $\mathbf{R}$ and $\mathbf{Y}$. Furthermore, the optimal value of $\alpha_X = \alpha$ can be recovered from $\alpha_R = \gamma_1$ by $\alpha = V\gamma_1$ and, since $\mathbf{X}\alpha = \mathbf{R} \gamma_1$, the canonical variates computed for $\mathbf{X}$ coincide with the ones computed for~$\mathbf{R}$.
\end{proof}

\section{Proof of PRCCA Kernel trick}
\label{supp:prcca:kernel}

\begin{lemma*}[PRCCA Kernel Trick]
The original PRCCA problem stated for $\mathbf{X}$ and $\mathbf{Y}$ can be reduced to solving the PRCCA problem for $\mathbf{R} = \begin{psmallmatrix}
\mathbf{R}_1\\
\mathbf{X}_2
\end{psmallmatrix}\in\mathbb{R}^{n + p_2}$ and $\mathbf{Y}$. 
The resulting canonical correlations and variates for these two problems coincide. The canonical coefficients for the original problem can be recovered via the linear transformation 
$\alpha_X = A\begin{psmallmatrix}
V_1 & 0\\
0&I
\end{psmallmatrix}\alpha_{R}.$
\end{lemma*}

\begin{proof}
To find the required linear transformation $A$ we first regress $\mathbf{X}_2$ from $\mathbf{X}_1$. We denote the matrix of regression coefficients by
$B = (\mathbf{X}_2^\top\mathbf{X}_2)^{-1}\mathbf{X}_2^\top\mathbf{X}_1\in\mathbb{R}^{p_2\times p_1}$
and set $A = \begin{psmallmatrix}
I & 0\\
-B & I
\end{psmallmatrix}.$ This transformation leads to the following transformed matrix 
$$\bwidetilde{\mathbf{X}} = \mathbf{X} A = (\bwidetilde{\mathbf{X}}_1, \bwidetilde{\mathbf{X}}_2) = \left(\mathbf{X}_1 - \mathbf{X}_2 B, \mathbf{X}_2\right).$$

It is easy to check that matrix $A$ is invertible and that the following relations hold  
\begin{center}
$A^{-1} = \begin{psmallmatrix}
I & 0\\
B & I
\end{psmallmatrix}~$
and  $~A^{-T}\begin{psmallmatrix}
I_{p_1} & 0\\
0 & 0
\end{psmallmatrix}A^{-1} = \begin{psmallmatrix}
I_{p_1} & 0\\
0 & 0
\end{psmallmatrix}.$    
\end{center}
Thus, if we denote $\widetilde\alpha = A^{-1}\alpha$, then
\begin{align*}
&\alpha^\top \frac {\mathbf{X}^\top\mathbf{Y}}{n}\beta =  \widetilde{\alpha}^\top \frac {\bwidetilde{\mathbf{X}}^\top\mathbf{Y}}{n}\beta\\
& \alpha^\top\left(\frac{\mathbf{X}^\top\mathbf{X}}{n}+\lambda_1 \begin{psmallmatrix}
I_{p_1} & 0\\
0 & 0
\end{psmallmatrix}\right)\alpha =\widetilde{\alpha}^\top\left(\frac{\bwidetilde{\mathbf{X}}^\top\bwidetilde{\mathbf{X}}}n+\lambda_1\begin{psmallmatrix}
I_{p_1} & 0\\
0 & 0
\end{psmallmatrix}\right)\widetilde{\alpha}
\end{align*}
The above equations imply that the PRCCA correlation coefficient (\ref{prcca:rho}) with $\lambda_2 = 0$ can be rewritten in terms of $\bwidetilde{\mathbf{X}}$ and~$\mathbf{Y}$
$$\rho_{PRCCA}(\widetilde\alpha, \beta; \lambda_1) = \dfrac{\widetilde\alpha^\top\widehat\Sigma_{\widetilde{X}Y}\beta}{\sqrt{\widetilde\alpha^\top\left(\widehat\Sigma_{\widetilde X\widetilde X}+\lambda_1 \begin{psmallmatrix}
I_{p_1} & 0\\
0 & 0
\end{psmallmatrix} 
\right)\widetilde\alpha}~\sqrt{\beta^\top\widehat\Sigma_{YY}\beta\vphantom{\begin{psmallmatrix}
I_{p_1} & 0\\
0 & 0
\end{psmallmatrix}}}}.$$

Next, let $\widetilde\alpha = \begin{psmallmatrix}
\scriptstyle\widetilde{\alpha}_1\\
\scriptstyle\widetilde{\alpha}_2
\end{psmallmatrix},$
where $\widetilde\alpha_1\in\mathbb{R}^{p_1}$ and $\widetilde\alpha_2\in\mathbb{R}^{p_2}$ correspond to blocks $\bwidetilde{\mathbf{X}}_1$ and $\widetilde{\mathbf{X}}_2$, respectively. Denote the orthogonal complement of $V_1$ by $V_1^\perp\in\mathbb{R}^{p_1\times (p_1 - n)}$ and 
consider the following transformation of the PRCCA coefficients 
\begin{center}
$\gamma_{1} = V_1^\top\widetilde\alpha_1$ and $\gamma_{2} =  \left(V_1^\perp\right)^\top\widetilde\alpha_1 $   
\end{center}
as well as the concatenation $\gamma = \begin{psmallmatrix} \gamma_{1}\\
\widetilde{\scriptstyle{\alpha}}_{2}
\end{psmallmatrix}$. Then, by analogy one can show that
\begin{align*}
&\widetilde\alpha^\top \frac{\bwidetilde{\mathbf{X}}^\top \mathbf{Y}}{n}\beta =\gamma_{1}^\top\frac{\mathbf{R}_1^\top\mathbf{Y}}{n}\beta + \widetilde\alpha_2^\top\frac{\bwidetilde{\mathbf{X}}_2^\top\mathbf{Y}}{n}\beta =  \gamma^\top\frac{\mathbf{R}^\top\mathbf{Y}}{n}\beta\\
&\widetilde\alpha^\top\left(\frac{\bwidetilde{\mathbf{X}}^\top\bwidetilde{\mathbf{X}}}{n} + \lambda_1 \begin{psmallmatrix}
I_{p_1} & 0\\
0 & 0
\end{psmallmatrix}\right)\widetilde\alpha
 =\gamma_1^\top\left(\frac{\mathbf{R}_1^\top\mathbf{R}_1}{n} + \lambda_1 I\right)\gamma_1 + \lambda_1 \gamma_2^\top\gamma_2+
\widetilde\alpha_2^\top\frac{\bwidetilde{ \mathbf{X}}_2^\top\bwidetilde{\mathbf{X}}_2}{n}\widetilde\alpha_2=\\
& =\gamma^\top\left(\frac{\mathbf{R}^\top\mathbf{R}}{n} + \lambda_1 \begin{psmallmatrix}
I_{p_1} & 0\\
0 & 0
\end{psmallmatrix}\right)\gamma+\lambda_1\gamma_2^\top\gamma_2
\end{align*}
where the last equation holds since $\mathbf{R}_1^\top \bwidetilde{\mathbf{X}}_2=  V_1^\top \widetilde{\mathbf{X}}_1^\top \widetilde{\mathbf{X}}_2 = 0$. 

Again, we can ignore  $\lambda_1\gamma_2^\top\gamma_2$ term as it is present only in the denominator of PRCCA correlation coefficient, which, therefore, can be rewritten in terms of $\mathbf{R}$ and $\mathbf{Y}$ as
$$\rho_{PRCCA}(\gamma, \beta;\lambda_1) = \frac{\gamma^\top\widehat\Sigma_{RY}\beta}{\sqrt{\gamma^\top\left(\widehat\Sigma_{RR}+\lambda_1 \begin{psmallmatrix}
I_{p_1} & 0\\
0 & 0
\end{psmallmatrix} 
\right)\gamma}~\sqrt{\beta^\top\widehat\Sigma_{YY}\beta\vphantom{\begin{psmallmatrix}
I_{p_1} & 0\\
0 & 0
\end{psmallmatrix} }}},$$
which is exactly the PRCCA correlation coefficient computed for $\mathbf{R}$ and $\mathbf{Y}$.
The optimal value of $\alpha_X = \alpha$ for the original problem can be recovered from $
\alpha_R = \gamma$ by 
$$\alpha = A\widetilde\alpha = A\begin{psmallmatrix}
V_1 \gamma_1\\
\scriptstyle\widetilde\alpha_2
\end{psmallmatrix} = A\begin{psmallmatrix}
V_1 & 0\\
0&I
\end{psmallmatrix}\gamma.$$ 
Moreover, 
$$\mathbf{X}\alpha = \bwidetilde{\mathbf{X}} \widetilde\alpha = 
(\mathbf{R}_1, \bwidetilde{\mathbf{X}}_2)\begin{psmallmatrix}
V_1^\top\widetilde{\alpha}_1\\
\scriptstyle\widetilde{\alpha}_2
\end{psmallmatrix} = \mathbf{R}\gamma,$$ so the canonical variates computed for $\mathbf{X}$ coincide with the ones computed for $\mathbf{R}$.
\end{proof}

\section{Proof of General RCCA to RCCA/PRCCA lemma}
\label{supp:genrcca:solution}

\begin{lemma*}[General RCCA to RCCA/PRCCA]
If both $K_X$ and $K_Y$ are positive definite then, by some proper change of basis, the general RCCA problem can be reduced to the RCCA one. Alternatively, if one of $K_X$ and $K_Y$ has zero eigenvalues then general RCCA boils down to solving the PRCCA problem with number of unpenalized coefficients equal to the multiplicity of zero eigenvalue.
\end{lemma*}

\begin{proof}
As usual, we do not penalize $Y$ part assuming $K_Y = 0$. Consider the eigendecomposition of matrix $K_X$, i.e. $K_X = UDU^\top$ with orthogonal $U\in\mathbb{R}^{p\times p}$ and diagonal $D\in\mathbb{R}^{p\times p}$. Since $K_X$ is supposed to be a positive semi-definite matrix, $D$ has non-negative diagonal elements. Applying transformation 
$\bwidetilde{\mathbf{X}} = \mathbf{X} U$ and $\widetilde \alpha = U^\top \alpha$, and using the same reasoning as in \ref{supp:rcca:kernel} and    \ref{supp:prcca:kernel} 
we obtain 
\begin{align*}
&\alpha^\top \frac {\mathbf{X}^\top\mathbf{Y}}{n}\beta =  \widetilde{\alpha}^\top \frac {\bwidetilde{\mathbf{X}}^\top\mathbf{Y}}{n}\beta\\
& \alpha^\top\left(\frac{\mathbf{X}^\top\mathbf{X}}{n}+K_X\right)\alpha =\widetilde{\alpha}^\top\left(\frac{\bwidetilde{\mathbf{X}}^\top\bwidetilde{\mathbf{X}}}n+D\right)\widetilde{\alpha}
\end{align*}
thus the equivalent modified correlation coefficient in the new basis is
\begin{align*}
    \rho(\widetilde\alpha, \beta; D) = \frac{\widetilde\alpha^\top\widehat\Sigma_{\widetilde{X}Y}\beta} {\sqrt{\widetilde\alpha^\top\left(\widehat\Sigma_{\widetilde X \widetilde X}+D\right)\widetilde \alpha}~\sqrt{\beta^\top\widehat\Sigma_{YY}\beta\vphantom{\left(\widehat\Sigma_{\widetilde X \widetilde X}+\lambda_1 D\right)}}}.
\end{align*}
Further, we decompose remaining diagonal matrix as $D =  S L S$ as follows. We have two cases:
\begin{enumerate}
    \item If all diagonal elements of $D$ are positive then we put $S =  D^{\frac12}$ and $L = I$.
    \item Suppose the first $p_1$ elements of $D$ are positive and the rest $p_2 = p - p_1$ elements are zero. Let $D = \begin{psmallmatrix}
    D_{11} & 0\\
    0 & 0
    \end{psmallmatrix}$, where $D_{11}\in\mathbb{R}^{p_1\times p_1}$ is the block containing all positive diagonal elements of matrix $D$.
    Then we can set
    \begin{center}
    $S = \begin{psmallmatrix}
    D_{11}^{\frac 12} & 0\\
    0 & I
    \end{psmallmatrix}$
    ~~and~~ $L = \begin{psmallmatrix}
    I_{p_1} & 0\\
    0 & 0
    \end{psmallmatrix}$ 
    \end{center}
\end{enumerate}
Note that, unlike $D$, matrix $S$ is non-singular, so the change of basis  
$\dbwidetilde{\mathbf{X}} = \bwidetilde{\mathbf{X}} S^{-1} $ and $\dswidetilde \alpha = S \widetilde \alpha$
is well-defined
and leads to the following equalities
\begin{align*}
&\widetilde{\alpha}^\top \frac {\bwidetilde{\mathbf{X}}^\top\mathbf{Y}}{n}\beta = \dswidetilde{\alpha}^\top \frac {\dbwidetilde{\mathbf{X}}^\top\mathbf{Y}}{n}\beta\\
&\widetilde{\alpha}^\top\left(\frac{\bwidetilde{\mathbf{X}}^\top\bwidetilde{\mathbf{X}}}n+D\right)\widetilde{\alpha} = \dswidetilde{\alpha}^\top\left(\frac{\dbwidetilde{\mathbf{X}}^\top\dbwidetilde{\mathbf{X}}}n+L\right)\dswidetilde{\alpha}.
\end{align*}
Therefore, the equivalent modified correlation coefficients 
$$\rho(\dswidetilde\alpha, \beta; L) = \frac{\dswidetilde\alpha^\top\widehat\Sigma_{\dwidetilde{X} \dwidetilde{X}}{\dswidetilde \alpha}}{\sqrt{\dswidetilde\alpha^\top\left(\widehat\Sigma_{\dwidetilde{X} \dwidetilde{X}} + L\right){\dswidetilde \alpha}}~\sqrt{\beta^\top\widehat\Sigma_{YY}\beta\vphantom{\left(\widehat\Sigma_{\dwidetilde{X} \dwidetilde{X}}\right)}}}.
$$
Is it easy to see that in the first case when $L = I$ 
the above correlation coefficient coincides with the RCCA correlation coefficient with $\lambda_1 = 1$. Alternatively, if $L =~\begin{psmallmatrix}
    I_{p_1} & 0\\
    0 & 0
    \end{psmallmatrix}$ 
it is equal to the PRCCA correlation coefficient with $\lambda_1 = 1$ and $p_2$ unpenalized coefficients.
Thus, we conclude that the General RCCA solution computed for $X$ and $Y$ can be found by means of either RCCA or PRCCA method applied to $\dwidetilde{X}$ and $Y$. The canonical variates stay the same regardless the basis as 
$\dbwidetilde{\mathbf{X}} \dswidetilde\alpha = \bwidetilde{\mathbf{X}}\widetilde\alpha = \mathbf{X}\alpha.$
The corresponding inverse transform for the coefficients is $\alpha = S U \dswidetilde\alpha$. 
\end{proof}

\section{Link between GRCCA and RCCA/PRCCA via the SVD of the penalty matrix}
\label{supp:grcca:solution}

Since GRCCA is just a special case of general RCCA, one can use previous lemma to map the GRCCA problem to either RCCA or PRCCA and, subsequesntly, find the canonical coefficients via the kernel trick. However, to find this transformation it is required to do an extra step: the eigendecomposition of the kernel matrices $K_X(\lambda_1, \mu_1)$ and $K_Y(\lambda_2, \mu_2)$. Although this can be infeasible in high dimensions for general penalty matrices it turns out that one can use the specific structure of the GRCCA penalty matrix to get around this eigendecomposition. 

We again assume that the regularization was imposed on the $X$ part only (i.e. $\lambda_2 =\mu_ 2 =~0$); however, it is not difficult to derive similar results for the general case. Recall that matrix $C_{m} = \frac{\mathbb{1}\mathbb{1}^\top}{m}\in\mathbb R^ {m\times m}$ has
\begin{itemize}
    \item unit eigenvalue with corresponding eigenvector $\frac{\mathbb{1}}{\sqrt{m}}\in\mathbb{R}^{m}$,
    \item zero eigenvalue with corresponding eigenspace $\left[\frac{\mathbb{1}}{\sqrt{m}}\right]^\perp\in~\mathbb{R}^{m\times(m-1)}$.
\end{itemize}
Here $[A]^\perp$ refers to the orthogonal complement of matrix $A$.
The resulting eigendecomposition is therefore 
\begin{center}
$C_{m} = U_{m}\begin{psmallmatrix}
1 & 0 \\
0 & 0
\end{psmallmatrix}U_{m}^\top$ ~with~ 
$U_{m} = \begin{psmallmatrix} \frac{\mathbb{1}}{\sqrt{m}} & \left[\frac{\mathbb{1}}{\sqrt{m}}\right]^\perp\end{psmallmatrix}
.$
\end{center}
It is easy to show the following eigendecomposition as well 
\begin{center}
$\lambda_1(I-C_{m}) + \mu_1 C_{m} = U_{m}D_{m}U_{m}^\top$ ~with~
$D_{m} = \begin{psmallmatrix}
\mu_1 & 0 \\
0 & \lambda_1 I_{m-1}
\end{psmallmatrix}.$
\end{center}
Thus the penalty matrix $K_X(\lambda_1,\mu_1)$ can be decomposed as
\begin{center}
$K_X(\lambda_1,\mu_1) = U D U^\top$ 
~with~ 
$U = U_{p_1}\oplus\ldots\oplus U_{p_K}$ 
~and~ 
$D = D_{p_1}\oplus\ldots\oplus D_{p_K}.$ 
\end{center}

Using the lemma from Section \ref{supp:genrcca:solution} we conclude that if $\lambda_1,\mu_1>0$ then the GRCCA problem can be solved via the RCCA approach. Alternatively, if $\lambda_1 = 0$ or $\mu_1 = 0$ it can be reduced to the PRCCA problem with $n-K$ and $K$ unpenalized coefficients, respectively. Hereafter we will assume $\lambda_1>0$, i.e. the presence of group homogeneity. 

Note that due to the specific structure of $K_X(\lambda_1,\mu_1)$ the eigendecomposition can be calculated block-wise. Moreover, the resulting computational cost is equal to the cost of computing the orthogonal complements $\left[\frac{\mathbb{1}}{\sqrt{p_1}}\right]^\perp,\ldots,\left[\frac{\mathbb{1}}{\sqrt{p_K}}\right]^\perp$, which can be efficiently found via, for example, Helmert contrasts (c.f. \texttt{contr.helmert()} function in R) and without computing any eigendecomposition.

\section{Link between GRCCA and RCCA/PRCCA via the feature matrix extension}

\label{supp:grcca:solution:alternative}

In Section \ref{supp:grcca:solution} we already proposed the way to solve GRCCA problem. If the number of groups $K$ is rather small then penalty matrix $K_X(\lambda_1,\mu_1)$ will consist of a few large blocks. In this case it would be quite expensive to compute the orthogonal complements and to further apply the linear transformation mapping General RCCA to RCCA/PRCCA. It turns out that there is an alternative linear transformation that leads to equivalent RCCA/PRCCA problem while being less expensive.

We need to establish some notation first. Suppose that matrix $\mathbf{X}$ is divided into blocks according to groups, i.e. $\mathbf{X} = \left(\mathbf{X}_1, \ldots, \mathbf{X}_K\right).$ Denote the column average matrix for group $k$ by $\mathbf{\bar{X}}_k = \mathbf{X}_k\frac{\mathbb{1}}{p_k}$ and 
consider an extended matrix $\bwidetilde{\mathbf{X}}$ that consists of (scaled) mean centered blocks and extra $K$ columns corresponding to (scaled) group means
$$\bwidetilde{\mathbf{X}}(a, b)  = \left(\sqrt{\frac{1}{a}}\left(\mathbf{X}_1 - \mathbf{\bar{X}}_1\right), \sqrt{\frac{p_1}{b}}\mathbf{\bar{X}}_1, \ldots, \sqrt{\frac{1}{a}}\left(\mathbf{X}_K - \mathbf{\bar{X}}_K\right), \sqrt{\frac{p_K}{b}}\mathbf{\bar{X}}_K\right)\in\mathbb{R}^{n\times(p+K)}.$$

\begin{lemma*}[GRCCA to RCCA/PRCCA] If $\mu_1>0$ the GRCCA problem for $\mathbf{X}$ and $\mathbf{Y}$ can be reduced to solving the RCCA problem for $\bwidetilde{\mathbf{X}}(\lambda_1,\mu_1)$ and $\mathbf{Y}$. If $\mu_1=0$, then GRCCA boils down to the PRCCA problem for $\bwidetilde{\mathbf{X}}(\lambda_1,1)$ and $\mathbf{Y}$ with $K$ unpenalized coefficients.
\end{lemma*}

\begin{proof}
Let us prove the statement for $K = 1$, one can easily extend the proof for arbitrary $K$.
For $K = 1$ the penalty matrix is
$$K_X(\lambda_1, \mu_1) = \lambda_1\left(I-\frac{\mathbb{1}\mathbb{1}^\top}{n}\right)+\mu_1\frac{\mathbb{1}\mathbb{1}^\top}{n}$$
which can be decomposed as
\begin{center}
$K_X(\lambda_1, \mu_1) = \widetilde U \widetilde D \widetilde U^\top$  ~with~ 
$\widetilde U = \left(I-\frac{\mathbb{1}\mathbb{1}^\top}{n}\frac{\mathbb{1}}{\sqrt{n}}
\right)$ ~and~ 
$\widetilde D=\begin{psmallmatrix}
\lambda_1 I_{n} & 0\\
0 & \mu_1\end{psmallmatrix}.$ 
\end{center}

This decomposition can be considered as an alternative to the eigendecomposition $K_X = UDU^\top$ discussed in the previous section. However, unlike matrix $U$ which was square and orthogonal, $\widetilde U$ is a rectangular $n\times(n+1)$ matrix with orthogonal rows, i.e. $\widetilde U \widetilde U^\top = I$. 

Further, similar to the previous section, we do decomposition $\widetilde D = \widetilde S \widetilde L \widetilde S$. We again have two cases here. 
If $\mu_1>0$ then we have
\begin{center}
$\widetilde S = \begin{psmallmatrix}
\sqrt{\lambda_1} I_{n} & 0\\
0 & \sqrt{\mu_1}\end{psmallmatrix}$ ~and~ $\widetilde L = I$ \end{center}
and the following relation is true
\begin{center}
$\mathbf{X}\widetilde U \widetilde S^{-1} = \left(\sqrt{\frac{1}{\lambda_1}}\left(\mathbf{X} - \mathbf{\bar{X}}\right), \sqrt{\frac{n}{\mu_1}}\mathbf{\bar{X}}\right) = \bwidetilde{\mathbf{X}}(\lambda_1, \mu_1)$.
\end{center}
If, on the contrary, $\mu_1=0$ then
\begin{center}
$\widetilde S = \begin{psmallmatrix}
\sqrt{\lambda_1} I_{n} & 0\\
0 & 1\end{psmallmatrix}$ ~and~ $\widetilde L = \begin{psmallmatrix}
I_{n} & 0\\
0 & 0\end{psmallmatrix}$
\end{center}
which leads us to  
\begin{center}
$\mathbf{X}\widetilde U \widetilde S^{-1} = \left(\sqrt{\frac{1}{\lambda_1}}\left(\mathbf{X} - \mathbf{\bar{X}}\right), \sqrt{n}\mathbf{\bar{X}}\right) = \bwidetilde{\mathbf{X}}(\lambda_1, 1)$.
\end{center}

Note that the proof in Section \ref{supp:genrcca:solution} requires matrix $\widetilde U$ to have orthogonal rows only. Thus, following this proof, one can conclude that the original GRCCA problem is equivalent to the RCCA problem solved for $\bwidetilde{\mathbf{X}}(\lambda_1, \mu_1)$ and $\mathbf{Y}$ if $\mu_1>0$. Alternatively, if $\mu_1=0$ then GRCCA boils down to solving the PRCCA problem for $\bwidetilde{\mathbf{X}}(\lambda_1, 1)$ and $\mathbf{Y}$, and one unpenalized coefficient. Again, such change of basis does not influence the canonical variates, whereas the canonical coefficients are transformed according to $\alpha = \widetilde U\widetilde S^{-1}  {\widetilde\alpha}$. 
\end{proof}

Note that the corresponding data transformation boils down to computing the group means and adjusting the feature matrix by group means. Although this approach can be considered as convenient alternative to the one suggested in Section \ref{supp:grcca:solution}, there is a trade-off. On the one hand, we reduce the cost by getting around the eigendecomposition (computing the orthogonal complements). On the other hand, we increase the feature matrix dimension from $p$ to $p+K$, which can be quite inefficient for large $K$.

\section{Neuroimaging analysis}

\label{supp:rde:preprocessing}

Details about the protocol and measures collected by HCP-DES are outlined in \cite{tozziHumanConnectomeProject2020}. Here, only the details relevant to this study are discussed.

\subsection{Neuroimaging acquisition details}

Images were acquired at the Stanford Center for Cognitive and Neurobiological Imaging (CNI) on a GE Discovery MR750 3 T scanner using a Nova Medical 32-channel head coil. Two spin-echo fieldmaps were acquired at the beginning of each session, one with a posterior-anterior phase encoding direction, the other with an anterior-posterior direction. All fMRI scans were conducted using a blipped-CAIPI simultaneous multislice “multiband” acquisition \citep{setsompopBlippedcontrolledAliasingParallel2012}.

\begin{enumerate}

\item Spin-echo fieldmaps: TE = 55.5 ms, TR = 6 s, FA = 90$\degree$, acquisition time = 18~s, field of view = 220.8 $\times$ 220.8 mm, 3D matrix size = 92 $\times$ 92 $\times$ 60, slice orientation = axial, angulation to anterior commissure - posterior commissure (AC-PC) line, phase encoding = AP and PA, receiver bandwidth = 250 kHz, readout duration = 49.14 ms, echo spacing~=~0.54~ms, voxel size = 2.4 mm isotropic.

\item Single-band calibration: TE = 30 ms, TR = 4.4 s, FA = 90$\degree$, acquisition time~=~13~s, field of view = 220.8 $\times$ 220.8 mm, 3D matrix size = 92 $\times$ 92~$\times$~60, slice orientation = axial, angulation to AC-PC line, phase encoding = PA, receiver bandwidth = 250 kHz, readout duration = 49.14 ms, echo spacing =~0.54~ms, number of volumes = 4, voxel size = 2.4 mm isotropic.

\item Multiband fMRI: TE = 30 ms, TR = 0.71 s, FA = 54$\degree$, acquisition time = 3:44 (Gambling task), field of view = 220.8 $\times$ 220.8 mm, 3D matrix size = 92 $\times$ 92 $\times$ 60, slice orientation = axial, angulation to AC-PC line, phase encoding = PA, receiver bandwidth = 250 kHz, readout duration = 49.14 ms, echo spacing = 0.54 ms, number of volumes for Gambling task = 316, multiband factor = 6, calibration volumes = 2, voxel size = 2.4 mm isotropic.

\item T1-weighted: TE = 3.548 ms, MPRAGE TR = 2.84 s, FA =~8$\degree$, acquisition time = 8:33, field of view = 256 $\times$ 256 mm, 3D matrix size = 320 $\times$ 320 $\times$~230, slice orientation = sagittal, angulation to AC-PC line, receiver bandwidth = 31.25 kHz, fat suppression = no, motion correction = PROMO, voxel size = 0.8 mm isotropic.

\item T2-weighted: TE = 74.4 ms, TR = 2.5 s, FA = 90$\degree$, acquisition time = 5:42, field of view = 240 $\times$ 240 mm, 3D matrix size = 320 $\times$ 320 $\times$ 216, slice orientation = sagittal, angulation to AC-PC line, receiver bandwidth = 125~kHz, fat suppression = no, motion correction = PROMO, voxel size = 0.8 mm isotropic.
\end{enumerate}

\subsection{Gambling task paradigm}

The HCP-DES adopts a version of the HCP Gambling task modified to allow comparison of small and large gain and loss outcomes \citep{somervilleLifespanHumanConnectome2018a, tozziHumanConnectomeProject2020}. A question mark is displayed on the screen and the participant must guess whether a number is greater than or less than five (and indicate their answer via button presses). If the participant identifies correctly, they win money, and if they guess incorrectly, they lose money. At the end of the task, 5 trials are randomly selected and summed together to determine the participant’s payment.

\subsection{Preprocessing}

Raw image files were converted to BIDS format and preprocessed using fMRIPrep  \citep{estebanFMRIPrepRobustPreprocessing2019}. Briefly, brain surfaces were reconstructed using recon-all (FreeSurfer 6.0.1, \cite{daleCorticalSurfaceBasedAnalysis1999}). Susceptibility distortion for fMRI data were corrected using the two echo-planar imaging (EPI) references with opposing phase-encoding directions  \citep{coxSoftwareToolsAnalysis1997}. Surface data was registered to fsaverage space and subcortical data to MNI space. These were then merged to grayordinate CIFTI files. Automated labeling of noise components following ICA decomposition was performed using AROMA  \citep{pruimICAAROMARobustICAbased2015a}. As final output, we down-sampled the preprocessed grey-ordinate functional CIFTI files to 32k FSLR space  \citep{glasserMinimalPreprocessingPipelines2013a}. Then, we applied a 4 mm full-width half-maximum smoothing constrained to the grey matter boundaries. For the quantification of brain responses to the task, the following conditions were convolved with a canonical hemodynamic response function as implemented in FSL  \citep{jenkinsonFSL2012}: high win, low win, high loss, low loss, high cue, low cue. The regressors obtained were entered in a design matrix together with the confound regressors generated by AROMA. A GLM analysis was then performed using the HCP pipelines  \citep{glasserMinimalPreprocessingPipelines2013a} and the contrast win $>$ loss was estimated for each participant (coefficients of high and low wins and losses were averaged). The z-scores corresponding to this contrast in each greyordinate were the features entered in the CCA analyses.


\bibliography{refs}

\end{document}